\documentclass[journal,onecolumn]{IEEEtran}

\usepackage[utf8]{inputenc}
\usepackage{cite}
\let\bf\undefined

\usepackage[hidelinks]{hyperref}
\usepackage{amsmath, amssymb, amsbsy}
\usepackage{amsthm}
\usepackage{nicematrix}
\usepackage{arydshln}
\usepackage{graphicx}
\usepackage{caption}
\newtheorem{definition}{Definition}
\newtheorem{corollary}{Corollary}
\newtheorem{theorem}{Theorem}
\newtheorem{lemma}{Lemma}
\newtheorem{remark}{Remark}

\newtheorem{conjecture}{Conjecture}
\newtheorem{proposition}{Proposition}
\usepackage[capitalize]{cleveref}
\newtheorem{claim}{Claim}
\usepackage{xcolor}
\usepackage[normalem]{ulem}
\usepackage{enumitem}
\usepackage{stmaryrd}
\usepackage{mathabx}
\usepackage{ulem}
\usepackage{multirow}
\usepackage[ruled,lined, linesnumbered]{algorithm2e}
\newif\ifcomment

\newcommand{\newinline}[1]{#1}
\newenvironment{new}{}

\newif\ifhighlightShortLongDiffs
\newif\ifversionShort
\newcommand{\versionShortLong}[2]{\ifversionShort \ifhighlightShortLongDiffs \textcolor{TUMRed}{Short version:} \fi #1\else\ifhighlightShortLongDiffs\textcolor{TUMRed}{Long version:}\fi #2 \ifhighlightShortLongDiffs \textcolor{TUMRed}{End of long version}\fi\fi}
\definecolor{TUMBlue}{RGB}{0,101,189} 
\definecolor{TUMBlueDark}{RGB}{0,82,147} 
\definecolor{TUMBlueLight}{RGB}{152,198,234} 
\definecolor{TUMBlueMedium}{RGB}{100,160,200} 

\definecolor{TUMIvory}{RGB}{218,215,203} 
\definecolor{TUMGreen}{RGB}{162,173,0} 
\definecolor{TUMGray}{gray}{0.6} 

\definecolor{TUMOrange}{RGB}{227,114,34} 
\definecolor{TUMGreenDark}{RGB}{0,124,48} 
\definecolor{TUMRed}{RGB}{196,7,27} 

\definecolor{TUMPink}{RGB}{181,92,165}
\definecolor{TUMPinkDark}{RGB}{155,70,141}
\definecolor{TUMPink1}{RGB}{198, 128, 187}
\definecolor{TUMPink2}{RGB}{214, 164, 206}
\definecolor{TUMPink3}{RGB}{230, 199, 225}
\definecolor{TUMPink4}{RGB}{246, 234, 244}
\usepackage[export]{adjustbox}
\usepackage[absolute, overlay]{textpos} 
\usepackage{pgfplots}
\pgfplotsset{compat=newest}
\usepackage{tikz}
\usetikzlibrary{external}
\usetikzlibrary{arrows, arrows.meta, shapes, shapes.multipart, trees, positioning, calc, backgrounds, fit, patterns, tikzmark}
\usetikzlibrary{matrix,decorations.pathmorphing,decorations.pathreplacing}
\usetikzlibrary{calligraphy}

\usepackage{hf-tikz}
\newcommand*\mymatrixbraceleft[4]{
    \draw[mymatrixbrace, decoration={raise=4pt}] (#1.west|-#1-#3-1.south west) -- node[left=4pt] {#4} (#1.west|-#1-#2-1.north west);
}
\newcommand*\mymatrixbraceright[4]{
    \draw[mymatrixbrace,decoration={raise=4pt}] (#1.east|-#1-#2-1.north east) -- node[right=7pt] {#4} (#1.east|-#1-#3-1.south east);
}
\newcommand*\mymatrixbracetop[4]{
    \draw[mymatrixbrace] (#1.north-|#1-1-#2.north west) -- node[above=2pt] {#4} (#1.north-|#1-1-#3.north east);
}
\newcommand*\mymatrixbracebottom[4]{
    \draw[mymatrixbrace, decoration={mirror, raise=1pt}] (#1.south-|#1-1-#2.south west) -- node[below=2pt] {#4} (#1.south-|#1-1-#3.north east);
}

\def\ve#1{{\mathchoice{\mbox{\boldmath$\displaystyle #1$}}%
              {\mbox{\boldmath$\textstyle #1$}}%
              {\mbox{\boldmath$\scriptstyle #1$}}%
              {\mbox{\boldmath$\scriptscriptstyle #1$}}}}

\newcommand{\0}{\ve{0}}
\newcommand{\F}{\ensuremath{\mathbb{F}}}
\newcommand{\Fq}{\ensuremath{\mathbb{F}_q}}
\newcommand{\Fqm}{\ensuremath{\mathbb{F}_{q^m}}}
\newcommand{\bbN}{\ensuremath{\mathbb{N}}}

\newcommand{\cC}{\ensuremath{\mathcal{C}}}

\newcommand{\cE}{\ensuremath{\mathcal{E}}}
\newcommand{\cG}{\ensuremath{\mathcal{G}}}
\newcommand{\cJ}{\ensuremath{\mathcal{J}}}

\newcommand{\cL}{\ensuremath{\mathcal{L}}}
\newcommand{\cM}{\ensuremath{\mathcal{M}}}
\newcommand{\cS}{\ensuremath{\mathcal{S}}}
\newcommand{\cV}{\ensuremath{\mathcal{V}}}
\newcommand{\cX}{\ensuremath{\mathcal{X}}}
\newcommand{\cZ}{\ensuremath{\mathcal{Z}}}

\newcommand{\ba}{\ve{a}}
\newcommand{\bb}{\ve{b}}
\newcommand{\bc}{\ve{c}}

\newcommand{\bf}{\ve{f}}

\newcommand{\bi}{\ve{i}}
\newcommand{\bm}{\ve{m}}
\newcommand{\bn}{\ve{n}}
\newcommand{\bu}{\ve{u}}

\newcommand{\bx}{\ve{x}}

\newcommand{\bA}{\ve{A}}
\newcommand{\bB}{\ve{B}}
\newcommand{\bC}{\ve{C}}

\newcommand{\bE}{\ve{E}}
\newcommand{\bG}{\ve{G}}
\newcommand{\bI}{\ve{I}}
\newcommand{\bM}{\ve{M}}
\newcommand{\bN}{\ve{N}}
\newcommand{\bP}{\ve{P}}
\newcommand{\bS}{\ve{S}}
\newcommand{\bT}{\ve{T}}
\newcommand{\bV}{\ve{V}}
\newcommand{\bX}{\ve{X}}
\newcommand{\bY}{\ve{Y}}
\newcommand{\G}{\ve{G}}

\newcommand{\bbeta}{\ve{\beta}}

\newcommand{\sfc}{\mathsf{c}}

\newcommand{\coloneqq}{\coloneq}

\newcommand{\myspan}[1]{\ensuremath{\left\langle #1\right\rangle}}

\newcommand{\ceil}[1]{\left\lceil#1\right\rceil}
\newcommand{\floor}[1]{\left\lfloor#1\right\rfloor}

\DeclareMathOperator{\diag}{diag}

\DeclareMathOperator{\rank}{rank}
\DeclareMathOperator*{\gcrd}{gcrd}
\DeclareMathOperator*{\lclm}{lclm}
\DeclareMathOperator{\rk}{rank}

\newcommand{\rkq}{\rk_q}
\newcommand{\wtSR}[1]{\ensuremath{\mathrm{wt}_{\mathsf{SR},#1}}}
\newcommand{\wtH}{\ensuremath{\mathrm{wt}_\mathsf{H}}}
\newcommand{\wtR}{\ensuremath{\mathrm{wt}_\mathsf{R}}}
\newcommand{\dSR}[1]{\ensuremath{\mathrm{d}_{\mathsf{SR},#1}}}
\newcommand{\dH}{\ensuremath{\mathrm{d}_\mathsf{H}}}

\newcommand{\extbasis}[1]{\mathrm{ext}_{#1}}

\newcommand{\Msym}{\ve{m}}
\newcommand{\mincut}{w}

\newcommand{\GLRS}{\G^{\sf (LRS)}}
\newcommand{\GLRSi}[1]{\G^{\sf (LRS)}_{#1}}
\newcommand{\alphas}[1]{\alpha_1,\dots,\alpha_{#1}}
\newcommand{\betalt}{\beta_{l,t}}
\newcommand{\betal}[1]{\beta_{l,#1}}

\newcommand{\PindSet}{\Omega}
\newcommand{\locSet}{\cL}
\newcommand{\rootSet}{\cZ}
\newcommand{\zeroSet}{Z}
\newcommand{\colMulVec}{\bb}

\newcommand{\aut}{\theta}

\newcommand{\FrobautPolyt}[2]{\Frobaut^{#2}(#1)}
\newcommand{\Frobaut}{\sigma}
\newcommand{\SkewVar}{X}

\newcommand{\FrobPolys}{\Fqm[\SkewVar;\Frobaut]}
\newcommand{\SkewPolys}{\Fqm[\SkewVar;\aut]}

\newcommand{\FrobPolysn}{\mulVarRng[\SkewVar;\Frobaut]}

\newcommand{\rowpoly}{f}

\newcommand{\priEle}{\gamma}

\newcommand{\indMap}{\varphi}
\newcommand{\indBlock}{l}

\newcommand{\evapt}[1]{\hat{#1}}

\newcommand{\numRows}{s}
\newcommand{\subRows}{\nu}
\newcommand{\numMulPolyVar}{n}
\newcommand{\mulVarRng}{R_\numMulPolyVar}
\newcommand{\fZt}{\mathsf{f}}
\newcommand{\tfzt}{\tau}
\newcommand{\Skewnk}{\cS_{n,k}}
\newcommand{\stepone}{\emph{(Step 1) }}
\newcommand{\steptwo}{\emph{(Step 2) }}
\newcommand{\matProne}{\emph{(I)}} 
\newcommand{\matPrtwo}{\emph{(II)}}

\newcommand{\up}{u}
\newcommand{\upc}{u}
\newcommand{\vp}{v}
\newcommand{\vpc}{v}

\commenttrue

\versionShortfalse 
\highlightShortLongDiffsfalse

\IEEEoverridecommandlockouts

\title{Linearized Reed-Solomon Codes with Support-Constrained Generator Matrix and Applications in Multi-Source Network Coding}

\author{
    Hedongliang Liu,~\IEEEmembership{Student Member,~IEEE},
    Hengjia Wei,\\
    Antonia Wachter-Zeh,~\IEEEmembership{Senior Member,~IEEE},
    Moshe Schwartz,~\IEEEmembership{Senior Member,~IEEE}%
    \thanks{The work has been supported in part by the German Research  Foundation (DFG) with a German Israeli Project Cooperation (DIP) under grants no.~PE2398/1-1, KR3517/9-1, and the major key project of Peng Cheng Laboratory under grant PCL2023AS1-2. The material in this paper was presented in part at the 2023 Information Theory Workshop (ITW).}%
    \thanks{Hedongliang Liu and Antonia Wachter-Zeh are with the School of Computation, Information and Technology, Technical University of Munich, Munich 80333, Germany (e-mail: lia.liu@tum.de;  antonia.wachter-zeh@tum.de).}%
    \thanks{Hengjia Wei is with the Peng Cheng Laboratory, Shenzhen 518055, China, with the School of Mathematics and Statistics, Xi'an Jiaotong University, Xi'an 710049, China, and with the Pazhou Laboratory (Huangpu), Guangzhou 510555, China (e-mail: hjwei05@gmail.com).}%
    \thanks{Moshe Schwartz is with the School of Electrical and Computer Engineering, Ben-Gurion University of the Negev, Beer Sheva 8410501, Israel (on a leave of absence), and with the Department of Electrical and Computer Engineering, McMaster University, Hamilton L8S 4K1, ON, Canada (e-mail: schwartz.moshe@mcmaster.ca).}%
}

\begin{document}

\maketitle

\begin{abstract}
  Linearized Reed-Solomon (LRS) codes are evaluation codes based on skew polynomials. They achieve the Singleton bound in the sum-rank metric and therefore are known as maximum sum-rank distance (MSRD) codes.
  In this work, we give necessary and sufficient conditions for the existence of MSRD codes with a support-constrained generator matrix. The \newinline{conditions on the support constraints} are identical to those for MDS codes and MRD codes.
  The required field size for an $[n,k]_{q^m}$ LRS codes with support-constrained generator matrix is $q\geq \ell+1$ and $m\geq \max_{l\in[\ell]}\{k-1+\log_qk, n_l\}$, where $\ell$ is the number of blocks and $n_l$ is the size of the $l$-th block. The special cases of the result coincide with the known results for Reed-Solomon codes and Gabidulin codes.
    For the support constraints that do not satisfy the necessary conditions, we derive the maximum sum-rank distance of a code whose generator matrix fulfills the constraints. 
    Such a code can be constructed from a subcode of an LRS code with a sufficiently large field size.
    Moreover, as an application in network coding, the conditions can be used as constraints in an integer programming problem to design distributed LRS codes for a distributed multi-source network.
\end{abstract}

\section{Introduction}

Designing error-correcting codes with support-constrained generator matrices is motivated by its application in weakly secure network coding for wireless cooperative data exchange \cite{yan2011weakly, song2012error, yan2014weakly, li2017cooperative}, where each node stores a subset of all messages and the nodes communicate via broadcast transmissions to disseminate the messages in the presence of an eavesdropper. It is closely related to designing codes with sparse and balanced generator matrices in wireless sensor networks \cite{dau2013balanced} and multiple access networks \cite{dau2014simple,halbawi2014distributed,halbawi2014distributedGab}, motivated by the balanced computation load during the encoding process while multiplying such a matrix \cite{halbawi2016balanced, halbawi2018sparse}.
From both, the theoretical and the practical point of view, the objective is to design codes with support-constrained generator matrix achieving the largest possible minimum distance.
\newinline{The support that we consider throughout this work is the Hamming support.}
In the Hamming metric, research focused on proving necessary and sufficient conditions for the existence of MDS codes fulfilling the support constraints. It was first conjectured in \cite{dau2014existence} (referred to as the \emph{GM--MDS conjecture}), further studied in \cite{heidarzadeh2017algebraic, yildiz2018further}, and finally proven by Yildiz and Hassibi \cite{yildiz2018optimum} and independently by Lovett \cite{lovett2018mds}.
\begin{theorem}[GM--MDS Condition {\cite{yildiz2018optimum,lovett2018mds}}]
  \label{thm:yildizMDScondition}
  Let $\zeroSet_1,\dots,\zeroSet_k\subseteq \{1,\dots, n\}$ be such that for any nonempty $\Omega\subseteq \{1,\dots, k\}$,
\begin{align}\label{eq:GM-MDScondition}
\left|\bigcap_{i\in\Omega}\zeroSet_i\right|+|\Omega|\leq k\ .
\end{align}
Then, for any prime power $q\geq n+k-1$, there exists an $[n,k]_{q}$ generalized Reed--Solomon (GRS) code with a generator matrix $\bG\in\Fq^{k\times n}$ fulfilling the support constraint:
\begin{align}\label{eq:zeroConstraints}
    \bG_{ij}=0\ ,\quad \forall i\in\{1,\dots,k\},\ \forall j \in \zeroSet_i\ .
\end{align}
\new{Moreover, if an MDS code has a generator matrix fulfilling the support constraint \eqref{eq:zeroConstraints}, then the sets $\zeroSet_i$'s satisfy \eqref{eq:GM-MDScondition}.}
\end{theorem}

Yildiz and Hassibi adapted the approach to Gabidulin codes in \cite{yildiz2019gabidulin} and derived the following GM--MRD condition.
\begin{theorem}[GM--MRD Condition {\cite[Theorem 1]{yildiz2019gabidulin}}]
  Let $\zeroSet_1,\dots,\zeroSet_k\subseteq \{1,\dots,n\}$ fulfill \eqref{eq:GM-MDScondition} for any nonempty $\Omega\subseteq \{1,\dots, k\}$.
Then, for any prime power $q$ and integer $m\geq \max \{n, k-1+\log_q k\}$, there exists an $[n,k]_{q^m}$ Gabidulin code with a generator matrix $\bG\in\Fqm^{k\times n}$ fulfilling \eqref{eq:zeroConstraints}.
\end{theorem}

\subsection{Related Work on Support-Constrained Generator Matrices}

Other variants of MDS codes with \newinline{a} support-constrained generator \newinline{matrix} have been studied. Apart from the condition in \eqref{eq:GM-MDScondition}, the work in \cite{greaves2019reed} considers the following two special cases of conditions on $\zeroSet_i$'s:
\begin{enumerate}
\item $\left|\bigcap_{i=1}^s\zeroSet_i\right|=k-s$ for all $s=1,\dots, k$. Note that when $s=k-1$, it is required that $\left|\bigcap_{i=1}^{k-1}\zeroSet_i\right|=1$, which means that there is at least one column of the generator matrix $\bG$ containing $k-1$ zeros. For this case an $[n,k]_q$ RS code generated by $\bG$ exists if $q\geq n$.
\item $\left|\zeroSet_i\right|\leq i-1$ for all $i=1,\dots,k$. Note that when $i\leq k-1$, $|\zeroSet_i|\leq k-2$, which implies less zeros are allowed in $\bG$ than \eqref{eq:GM-MDScondition}. For this case an $[n,k]_q$ RS code generated by $\bG$ exists if $q\geq n+1$.
\end{enumerate}

In \cite{song2018generalized,chen2022sparse}, the existence of an MDS code with a sparse and balanced generator matrix $\bG$ was studied. ``Sparse'' means that each row of $\bG$ has the maximum number of zeros, i.e., $k-1$ zeros, and ``balanced'' means that the number of zeros in any two columns differs by at most one, i.e., the weight of each column is either $\ceil{k(n-k+1)/n}$ or $\floor{k(n-k+1)/n}$.
It is shown in \cite{song2018generalized} that for any $1\leq k\leq n$, if $q\geq n+\ceil{k(k-1)/n}$, then there exists an $[n,k]_q$ generalized RS code with a sparse and balanced generator matrix.
More recently, \cite{chen2022sparse} showed that for any $k\geq 3$, $\frac{k}{n}\geq \frac{1}{2}$, and $q\geq n-1$, there exists an $[n,k]_q$ MDS code with a sparse and balanced generator matrix.

\begin{new}
It was shown in \cite{roth1985generator} that for any $q\geq n$, every $[n,k]_q$ GRS code has a systematic generator matrix $\bG_{\sf sys}=
\begin{pmatrix}
\bI|\bA
\end{pmatrix}
$ where $\bA$ is a (generalized) Cauchy matrix, and conversely, any matrix in the form of $
\begin{pmatrix}
\bI|\bA
\end{pmatrix}
$ where $\bA$ is a (generalized) Cauchy matrix generates a GRS code. This implies the following result.
\begin{proposition}\label{prop:systematric_field_size}
  Let $\zeroSet_1,\dots,\zeroSet_k\subseteq \{1,\dots, n\}$ be such that for any nonempty $\Omega\subseteq \{1,\dots, k\}$,
\begin{align}
  \label{eq:systematic_zero_constraints}
\left|\bigcap_{i\in\Omega}\zeroSet_i\right|+|\Omega|=k\ .
\end{align}
Then, for any prime power $q\geq n$,
there exists an $[n,k]_q$ GRS with a generator matrix $\bG$ fulfilling \eqref{eq:zeroConstraints}. 
\end{proposition}
\begin{proof}
  The goal is to show that for any $\bG\in\Fq^{k\times n}$ with zeros fulfilling \eqref{eq:systematic_zero_constraints}, there is $k\times k$ submatrix forming a generalized permutation matrix $\bP$
  (i.e., a permutation matrix where the ones can be replaced by any nonzero element in the field). Denote $C_{\Omega}\coloneqq \bigcap_{i\in\Omega}\zeroSet_i$. Note that $C_{\Omega}\subset[n]$ and $|C_{\Omega}|\leq k$.
  \begin{itemize}
      \item 
  For any $\Omega$ such that $|\Omega|=k$, $|C_{\Omega}|=0$. This means there is no zero column in $\bG$.
  \item 
  For any $\Omega$ such that $|\Omega|=k-1$, $|C_{\Omega}|=1$. This means there is at least one column in $\bG$ has $k-1$ zeros. Assume for some $\Omega_1\neq \Omega_2$ of size $k-1$ , $C_{\Omega_1}=C_{\Omega_2}$. Then for $\Omega=\Omega_1\cup\Omega_2$ with $|\Omega|=k$, $|C_{\Omega}|=1$, which contradicts the case above. Hence, we conclude that for every distinct $\Omega$ of size $k-1$, there is a distinct $C_{\Omega}$ of size one, i.e., a unique column with $k-1$ zeros. There are $\binom{k}{k-1}$ such $\Omega$, therefore, there are $k$ distinct columns with $k-1$ zeros.
  \end{itemize} 
  Now we will show that for any two columns with $k-1$ zeros, the nonzero elements are not in the same row.
  Assume the opposite, i.e., there are two columns with $k-1$ zeros, whose nonzero elements are in the same row (equivalently, whose zeros are in the same rows). Let $\Omega_0$ be the rows where the zeros of two columns are. Then $|\Omega_0|=k-1$ and $|C_{\Omega_0}|=2$, which contradicts~\eqref{eq:systematic_zero_constraints}.
\end{proof}

For codes in rank metric, it has been shown in \cite[Th.~4.4]{neri2020systematic} that
for any $n\leq m$, every $[n,k]_{q^m}$ Gabidulin code has a systematic generator matrix $\bG_{\sf sys}=
\begin{pmatrix}
\bI|\bA
\end{pmatrix}$ where $\bA$ is the $q$-Cauchy matrix, and conversely, any matrix in the form of $
\begin{pmatrix}
\bI|\bA
\end{pmatrix}
$, where $\bA$ is a $q$-Cauchy matrix, generates a Gabidulin code.
Similar to \cref{prop:systematric_field_size}, the following result can be derived for MRD codes.
\begin{proposition}\label{prop:systematric_field_size_Gab}
  Let $\zeroSet_1,\dots,\zeroSet_k\subseteq \{1,\dots, n\}$ fulfill \eqref{eq:systematic_zero_constraints} for any nonempty $\Omega\subseteq \{1,\dots, k\}$.
Then,
for any prime power $q$ and integer $m\geq n$, there exists an $[n,k]_{q^m}$ Gabidulin code with a generator matrix $\bG$ fulfilling \eqref{eq:zeroConstraints}.
\end{proposition}
We conjecture that the smaller sufficient field sizes also holds for the general case. 
\begin{conjecture}\label{cor:smaller_field_size_general}
  Let $\zeroSet_1,\dots,\zeroSet_k\subseteq \{1,\dots, n\}$ fulfill \eqref{eq:GM-MDScondition} for any nonempty $\Omega\subseteq \{1,\dots, k\}$.
Then, for any prime power $q\geq n$,
there exists an $[n,k]_q$ GRS with a generator matrix fulfilling \eqref{eq:zeroConstraints}. 
Moreover, for any prime power $q$ and integer $m\geq n$, there exists an $[n,k]_{q^m}$ Gabidulin code with a generator matrix fulfilling \eqref{eq:zeroConstraints}.
\end{conjecture}
\end{new}

\subsection{Related Work in the Sum-Rank Metric}
The sum-rank metric was first considered in coding for MIMO block-fading channels \cite{el2003design, lu2005unified} and the design of AM-PSK constellations \cite{lu2006constructions}. It was then explicitly introduced in multi-shot network coding \cite{nobrega2009multishot}.
The \emph{minimum sum-rank distance} is a direct analogue to \emph{transmit diversity gain} and the \emph{maximum sum-rank distance} property is a direct analogue to \emph{rate-diversity optimality}.
An explicit construction of optimal space-time codes from sum-rank metric codes over a finite field was first given in \cite{shehadeh2021space}. Additionally, sum-rank metric codes have been considered in applications such as network streaming \cite{mahmood2016convolutional}, distributed storage systems \cite{martinez2019universal,martinez2020locally,cai2021construction} and post-quantum secure code-based cryptosystems \cite{dalconzo2022codeequivalence,hormann2022security}.

Linearized Reed-Solomon (LRS) codes \cite{liu2015construction, martinez2018skew} are a class of evaluation codes based on skew polynomials \cite{ore1933theory}, achieving the Singleton bound in the sum-rank metric, and therefore known as maximum sum-rank distance (MSRD) codes. They have been applied in network coding \cite{martinez2019reliable}, locally repairable codes \cite{martinez2019universal} and code-based cryptography \cite{hormann2022security}.
  Extensive research has been done in recent years in the subareas of fundamental coding-theoretical properties of sum-rank metric codes, e.g., \cite{byrne2020fundamental, ott2021bounds, martinez2019theory,camps2022optimal, ott2022covering}, constructions of perfect/optimal/systematic sum-rank metric codes \cite{alfarano2022sum, martinez2020hamming, almeida2020systematic, martinez2022generalMSRD, caruso2022duals} and decoding algorithms for LRS codes, e.g.,
\cite{boucher2020algorithm,puchinger2021bounds,bartz2021fast,puchinger2022generic,Hoermann2022efficient,Hoermann2022errorerasure,jerkovits2022universal,bartz2022fast}.

\subsection{Our Contribution}
Motivated by practical interest in codes with support-constrained generator matrices and research on sum-rank metric codes (in particular, LRS codes), we investigate the existence of MSRD codes with a support-constrained generator matrices in this work.
As a result, we show that \eqref{eq:GM-MDScondition} is also the necessary and sufficient conditions for the existence of MSRD codes with a generator matrix fulfilling \eqref{eq:zeroConstraints}.
Further, our main result is that for any prime power $q\geq\ell+1$ and integer $m\geq \max_{l\in[\ell]}\{k-1+\log_qk, n_l\}$, there exists a linearized Reed-Solomon code in $\Fqm$ of length $n$, dimension $k$, and with $\ell$ blocks, such that its generator matrix satisfies the support constraints.
Moreover, if the conditions in \eqref{eq:GM-MDScondition} are not satisfied, we show that the largest possible sum-rank distance can be achieved by subcodes of LRS codes.
For a distributed multi-source network as introduced in \cite{halbawi2014distributed, halbawi2014distributedGab}, where a supported-constrained generator matrix is required for reliable communication against malicious (or failed) nodes in the network, we introduce a scheme to design a \emph{distributed LRS code} for any network instance. The scheme illustrates how the necessary and sufficient conditions derived in this work can be used as constraints in a linear programming problem to design the parameters of a desired distributed LRS code.

The rest of the paper is organized as following: In \cref{sec:preliminary} we fix the notations used throughout the paper and give the preliminaries on sum-rank metric, skew polynomials and LRS codes that are needed to prove the main results.
We present in \cref{sec:field_size} our main result on the necessary and sufficient conditions for the existence of MSRD codes with a support-constrained generator matrix as in \eqref{eq:zeroConstraints} and the sufficient field size of an LRS code fulfilling the support constraint.
  In \cref{sec:app_network_cod} we introduce the distributed multi-source network model and present a scheme to design a distributed LRS code for this network model, which uses the necessary and sufficient conditions derived in \cref{sec:field_size}.
\cref{sec:main_result} provides the proof of the main result, which is presented as a claim in \cref{sec:field_size}.


\section{Preliminaries}
\label{sec:preliminary}

\subsection{Notations}
Denote by $[a,b]$ the set of integers $\{a, a+1,\dots, b-1, b\}$, and let $[b]:=[1,b]$. 
Let $\bbN$ be the set of natural numbers and $\bbN_0\coloneqq\bbN\cup \{0\}$.
Denote by $\F_q$ the finite field of size $q$, and by $\Fqm$ its extension field of extension degree $m$.
  Fix a\newinline{n ordered} basis $\bbeta=(\beta_1,\beta_2,\dots,\beta_m)$ of $\Fqm$ over $\Fq$. We define a mapping from $\Fqm^n$ to $\Fq^{m\times n}$ by
  \begin{align}
    \extbasis{\bbeta} : \Fqm^n&\to\Fq^{m\times n} \nonumber\\
    \bc =(c_1,c_2,\dots, c_n) &\mapsto \bC=
                                \begin{pmatrix}
                                  c_{1,1} & c_{1,2} & \dots & c_{1,n}\\
                                  \vdots & \vdots & \ddots & \vdots \\
                                  c_{m,1} & c_{m,2} & \dots & c_{m,n}
                                \end{pmatrix} \label{eq:ext-map}
  \end{align}
  where $\bC$ is unique such that $c_j=\sum_{i=1}^{m}c_{i,j}\beta_i$, for all $j=1,\dots, n$.
  The $\Fq$-rank of $\bc$ is defined as $\rank_q(\bc)\coloneq \rank(\bC)$.

Given two vectors $\ba=(a_1,\dots,a_n),\bb=(b_1,\dots, b_n)\in\F^n$, we define their \emph{star-product} as the \newinline{element}-wise multiplication, i.e.,
$\ba\star\bb\coloneqq (a_1b_1, a_2b_2, \dots, a_nb_n)\in\F^n$.
Given a vector $\ba\in\F^n$, let $\diag(\ba)\in\F^{n\times n}$ be the diagonal matrix with entries of $\ba$ on its diagonal.

Throughout the paper, unless specified otherwise, the indices of entries in vectors, elements in sets, etc., start from $1$, while the coefficients of polynomials start from $0$.

\subsection{Sum-Rank Metric}
Let $\ell\in\bbN$ and $\bn_\ell= (n_1, \dots,n_\ell)\in\bbN^\ell$ be an ordered partition of $n=\sum_{\indBlock=1}^\ell n_\indBlock$.
For a matrix 
\begin{align*}
  \bA =
  \begin{pmatrix}
    \bA_1 & \bA_2 & \dots & \bA_\ell
  \end{pmatrix}
                        \in\Fq^{m\times n}
                        \quad\quad\text{ or }\quad\quad
                        \bB =
                        \begin{pmatrix}
                          \bB_1\\
                          \bB_2\\
                          \vdots\\
                          \bB_\ell
                        \end{pmatrix}
  \in\Fq^{n\times m}
\end{align*}
where $\bA_i\in\Fq^{m\times n_i}$ and each $\bB_i\in\Fq^{n_i\times m}$, we say $\bA$ has a \emph{column-wise} ordered partition with respect to $\bn_{\ell}$ and $\bB$ has a \emph{row-wise} ordered partition w.r.t.~$\bn_{\ell}$. The sum-rank weight w.r.t.~$\bn_\ell$ is respectively
\begin{align*}
  \wtSR{\bn_\ell}(\bA)= \sum_{\indBlock=1}^{\ell} \rank(\bA_\indBlock)
  \quad\quad\text{ or }\quad\quad
  \wtSR{\bn_\ell}(\bB)= \sum_{\indBlock=1}^{\ell} \rank(\bB_\indBlock)
\end{align*}
We can find the following relation between the sum-rank weight and the rank of a matrix.
\begin{lemma}
  \label{lem:rank-leq-sumrank-mat}
  For a matrix $\bA\in\Fq^{m\times n}$ and an ordered partition $\bn_\ell= (n_1, \dots,n_\ell)$ of $n$, $\rank(\bA)\leq\wtSR{\bn_{\ell}}(\bA)\leq \ell\cdot \rank(\bA)$. Similarly, for a matrix $\bB\in\Fq^{n\times m}$, $\rank(\bB)\leq \wtSR{\bn_{\ell}}(\bB)\leq \ell\cdot \rank(\bB)$.
\end{lemma}
\begin{proof}
  Denote by $\myspan{\bA}_{\sfc}$ the column space of a matrix $\bA$. For the first inequality,
  \begin{align*}
    \rank(\bA) = \dim(\myspan{\bA}_{\sfc}) &= \dim(\myspan{\bA_1}_{\sfc}+\cdots + \myspan{\bA_\ell}_{\sfc})\\
                                       &\leq \dim(\myspan{\bA_1}_{\sfc}) + \cdots \newinline{+}\dim(\myspan{\bA_\ell}_{\sfc})\\
    &=\rank(\bA_1) +\cdots +\rank(\bA_\ell) = \wtSR{\bn_\ell}(\bA)\ .
  \end{align*}
  For the second inequality,
  \begin{align*}
    \wtSR{\bn_\ell}(\bA)= \sum_{i=1}^{\ell}\rank(\bA_i) 
    \leq \sum_{i=1}^{\ell}\rank(\bA) =\ell\cdot \rank(\bA).
  \end{align*}
  For the matrix $\bB\in\Fq^{n\times m}$ with a row-wise ordered partition, the proof is similar with considering the row space of $\bB$ and $\bB_i$'s.
\end{proof}
For the vector space $\Fqm^n$, with an $\ell$-ordered partition $(n_1, \dots,n_\ell)\in\bbN^\ell$, the sum-rank metric is defined as the following. For simplicity of the notation, we abuse the notation $\wtSR{\bn_\ell}$ for the matrix space over $\Fq$ and the vector space over $\Fqm$.

\begin{definition}
The sum-rank weight on $\Fqm^n$, with an ordered partition $\bn_\ell =(n_1, \dots,n_\ell)$ of $n$, is defined as
\begin{align*}
    \wtSR{\bn_\ell}(\cdot)\ :\ \Fqm^n &\to \bbN_0\\
    \bx &\mapsto \sum_{i=1}^\ell \rkq(\bx_i)
\end{align*}
where $\bx = (\bx_1 | \bx_2 | \dots | \bx_\ell)$ with $\bx_i\in\Fqm^{n_i}$. Moreover, the sum-rank distance is defined as
\begin{align*}
    \dSR{\bn_\ell}(\cdot,\cdot)\ :\ \Fqm^n\times\Fqm^n &\to \bbN_0\\
    (\ba,\bb) &\mapsto \wtSR{\bn_\ell}(\ba-\bb).
\end{align*}
For a linear code $\cC\subseteq\Fqm^n$, its minimum sum-rank distance is
\begin{align*}
  \dSR{\bn_\ell}(\cC)\coloneqq \min_{\bc_1,\bc_2\in\cC} \dSR{\bn_\ell}(\bc_1,\bc_2) = \min_{\0\neq \bc\in\cC} \wtSR{\bn_\ell}(\bc).
\end{align*}
\end{definition}
It is well known that for $\ell=1$ the sum-rank \newinline{metric} coincides with the rank metric and for $\ell=n$, it is the Hamming metric (see also \cite[Proposition 1.4, 1.5]{FnTsurvey-Umberto}).
The following lemma gives a relation among \newinline{the weights of a vector $\bx\in\Fqm^n$ in} the Hamming metric, the sum-rank metric and the rank metric.

\begin{lemma}[\hspace{1pt}\cite{ott2022covering}]
  \label{lem:sum-rank-Hamming}
  For a vector $\bx\in\Fqm^n$ and any ordered partition $\bn_\ell =(n_1, \dots,n_\ell)$ of $n$, $\wtR(\bx)\leq \wtSR{\bn_{\ell}}(\bx)\leq \wtH(\bx)$, where $\wtR(\bx)=\rank_q(\bx)$ is the weight of $\bx$ in the rank metric and $\wtH(\bx)$ is the weight of $\bx$ in the Hamming metric. 
\end{lemma}

\subsection{Skew Polynomials}
Denote by $\Fqm[x]$ a univariate (commutative) polynomial ring with coefficients from $\Fqm$ and by $\Fqm[x_1,\dots, x_\numMulPolyVar]$ a multivariate polynomial ring for $\numMulPolyVar\geq 1$.
The following lemma gives \newinline{a} sufficient \newinline{condition on the} field size such that a nonzero evaluation of a multivariate polynomial \newinline{in the field} always exists.
\begin{lemma}[Combinatorial Nullstellensatz~{\cite[Theorem 1.2]{alon1999combinatorial}}]
  \label{lem:nullstellensatz}
  Let $\F$ be an arbitrary field and $f$ be a nonzero polynomial in $\F[x_1,\dots,x_\numMulPolyVar]$ of total degree $\deg(f)=\sum_{i=1}^n t_i$, where $t_i\geq 0, \ \forall i$.
  Then, if $\cX_1,\dots, \cX_\numMulPolyVar$ are subsets of $\F$ with $|\cX_i|> t_i$, then there are $\evapt{x}_1\in \cX_1,\dots, \evapt{x}_\numMulPolyVar\in \cX_\numMulPolyVar$ so that
  \begin{align*}
    f(\evapt{x}_1,\dots,\evapt{x}_\numMulPolyVar)\neq 0.
  \end{align*}
\end{lemma}
Let $\SkewPolys$ be a skew polynomial ring over $\Fqm$ with automorphism $\aut: \Fqm\to\Fqm$.
The degree of a skew polynomial $f(\SkewVar) = \sum_i f_i \SkewVar^{i} \in\SkewPolys$ is $\deg f(\SkewVar)\coloneqq\max\{i ~|~ f_i\neq 0\}$ \newinline{for $f(\SkewVar)\neq 0$, and $\deg f(\SkewVar)=-\infty$ if $f(\SkewVar)=0$}.
The addition in $\SkewPolys$ is defined to be the usual addition of polynomials and the multiplication is defined by the basic rule $\SkewVar\cdot \alpha =\aut(\alpha)\cdot \SkewVar, \forall \alpha\in \Fqm$ and extended to all elements of $\SkewPolys$ by associativity and distributivity.
For two skew polynomials $f(\SkewVar)=\sum_i f_i\SkewVar^{i}$ and $g(\SkewVar)=\sum_j g_j\SkewVar^{j}$, their product is
\begin{align}\label{eq:skewProd}
f(\SkewVar)\cdot g(\SkewVar)=\sum_i\sum_j f_i\aut^i(g_j) \SkewVar^{i+j}.
\end{align}
The degree of the product is $\deg \left(f(\SkewVar)\cdot g(\SkewVar)\right) = \deg f(\SkewVar) + \deg g(\SkewVar)$.
For ease of notation, when it is clear from the context, we may omit the variable notation in $f(\SkewVar)$ for $f\in\SkewPolys$, and write only $f$.

Since skew polynomials are non-commutative under multiplication and division, we denote by $|_r$ and $|_l$ the right and left \newinline{divisibility} respectively.
The powers of $\aut$ are $\aut^i(\alpha)=\aut(\aut^{i-1}(\alpha))$. 
For any $\alpha\in\Fqm$, its \emph{$i$-th truncated norm} is defined as $N_i(\alpha)\coloneqq\prod_{j=0}^{i-1} \aut^j(\alpha)$ and $N_0(\alpha)=1$.
For the Frobenius automorphism, $\Frobaut: \alpha\mapsto \alpha^q$, $\Frobaut^i(\alpha) = \alpha^{q^i}$, and $N_i(\alpha) = \alpha^{(q^i-1)/(q-1)}$.
\begin{definition}[$\aut$-Conjugacy Classes]
  \label{def:conjugacyClasses}
  Two elements $a,b\in\Fqm$ are called $\aut$-\emph{conjugate} if there exists a nonzero element $c\in\Fqm$ such that $b=\aut(c)a c^{-1}$. Otherwise, they are called $\aut$-\emph{distinct}.
  The \emph{conjugacy class} of $a$ with respect to $\aut$ is the set
  \begin{align*}
    C_{\aut}(a)\coloneq \{ \aut(c)a c^{-1} ~|~ c\in \Fqm\setminus\{0\} \}\ .
  \end{align*}
\end{definition}

\begin{theorem}[Structure of $\Frobaut$-conjugacy classes {\cite[Theorem 2.12]{FnTsurvey-Umberto}}]
  \label{thm:conjugacyClasses}
  Let $\priEle$ be a primitive element of $\Fqm$. For the Frobenius automorphism $\Frobaut$, the $q-1$ elements $1,\priEle,\priEle^2,\dots, \priEle^{q-2}$ are pair-wise $\Frobaut$-distinct.
  There are exactly $q-1$ distinct nonzero $\Frobaut$-conjugacy classes, each of size $\frac{q^m-1}{q-1}$, in $\Fqm$, and
  $$\Fqm=C_{\Frobaut}(0)\cup C_{\Frobaut}(\priEle^0)\cup\cdots \cup C_{\Frobaut}(\priEle^{q-2}),$$ where the union is disjoint.
\end{theorem}

\begin{definition}[Remainder Evaluation of Skew Polynomials~\cite{ore1933theory}]
\label{def:eva_division}
For $f(\SkewVar)\in\SkewPolys$, $\alpha\in\Fqm$, since division on the right is possible for any $f(\SkewVar)\in \SkewPolys$, we may write $f(\SkewVar)=q_r(\SkewVar)(\SkewVar-\alpha)+t$, with $t\in\Fqm$. The (right) evaluation of $f(\SkewVar)$ is then defined as
\begin{align*}
    f(\alpha) = t.
\end{align*}
\end{definition}
The next lemma shows that the remainder evaluation can be computed without using the division algorithm and it is equivalent to the evaluation in \cref{def:eva_division}. A proof can also be found in \cite[Theorem 2.3]{baumbaugh2016results}.
\begin{lemma}[Explicit Expression of Remainder Evaluation of Skew Polynomials~{\cite[Lemma 2.4]{lam1988vandermonde}}]\label{lem:evaluation}
For $f=\sum_{i=0}^{\deg f} f_i \SkewVar^{i}\in \SkewPolys$ and $\alpha\in\Fqm$, $f(\alpha) = \sum_{i=0}^{\deg f}f_iN_{i}(\alpha)$. In particular, if $\aut$ is the Frobenius automorphism $\Frobaut$, $f(\alpha) = \sum_{i=0}^{\deg f} f_i \alpha^{(q^{i}-1)/(q-1)}$.
\end{lemma}
Similar to the evaluation of conventional polynomials, the evaluation of a $f\in\SkewPolys$ at $\Omega=\{\alphas{n}\}\subseteq\Fqm$ can be written as $(f(\alpha_1), \dots, f(\alpha_n)) = \bf\cdot \bV_k^\aut(\Omega)$, where $k$ is the degree of $f$, $\bf=(f_0,\dots,f_{k})$ contains the coefficients of $f$, and $\bV_{k+1}^\aut(\Omega)$ is the first $k+1$ rows of $\bV^\aut(\Omega)$ defined below.
\begin{definition}[$\aut$-Vandermonde Matrix]
Let $\aut$ be an automorphism $\aut:\Fqm\to\Fqm$. Given a set $\PindSet=\{\alphas{n}\} \subseteq \Fqm$,
the $\aut$-Vandermonde matrix of $\PindSet$ is given by
\begin{align*}
    \bV^\aut(\PindSet) \coloneqq
    \begin{pmatrix}
      N_0(\alpha_1) & N_0(\alpha_2) & \dots & N_0(\alpha_n)\\
      N_1(\alpha_1) & N_1(\alpha_2) & \dots & N_1(\alpha_n)\\
      \vdots & \vdots & \ddots &\vdots \\
      N_{n-1}(\alpha_1) & N_{n-1}(\alpha_2) & \dots & N_{n-1}(\alpha_n)
    \end{pmatrix}\ ,
\end{align*}
where $N_0(\alpha)=1$, and for $i\geq 1$, $N_i(\alpha)=\prod_{j=0}^{i-1} \aut^j(\alpha)$ is the $i$-th truncated norm of $\alpha$.
\end{definition}

\begin{definition}[Minimal Polynomial]\label{def:minPoly}
Given a nonempty set $\PindSet\subseteq \Fqm$, we \newinline{say $f_{\PindSet}$ is a \emph{minimal polynomial} of $\PindSet$ if it is} a monic polynomial of minimal degree such that $f_{\PindSet}(\alpha)=0$ for all $\alpha\in\PindSet$.
\end{definition}
It was shown in \cite[Lemma 5]{lam1986general} (see also \cite[Theorem 2.5]{FnTsurvey-Umberto}) that the minimal polynomial of any nonempty set $\Omega=\{\alphas{n}\}\subseteq \Fqm$ is unique.
The minimal polynomial can be constructed by an iterative Newton interpolation approach
\versionShortLong{as in \cite[Proposition 2.6]{FnTsurvey-Umberto} or by computing}{as follows (see also \cite[Proposition 2.6]{FnTsurvey-Umberto}):
First, set
\begin{align*}
  g_1 =\SkewVar-\alpha_1 .
\end{align*}
Then for $i=2,3,\dots, n$, perform
\begin{align}\label{eq:newtonInterpolation}
  g_i=
  \begin{cases}
    g_{i-1}\quad &\text{if } g_{i-1}(\alpha_i)=0,\\
    \left(\SkewVar-\alpha^{g_{i-1}(\alpha_i)}\right)\cdot g_{i-1}\quad &\text{otherwise,}
  \end{cases}
\end{align}
where $\alpha^{g_{i-1}(\alpha_i)}\coloneq \aut(g_{i-1}(\alpha_i))\alpha_i g_{i-1}(\alpha_i)^{-1}$ is the $\aut$-conjugate of $\alpha$ w.r.t.~$g_{i-1}(\alpha_i)$.
Upon termination, $g_n(\SkewVar)=f_{\Omega}(\SkewVar)$.

It can be seen that the minimal polynomial of a set $\Omega$ can be also constructed by computing}
\begin{align}\label{eq:lclm}
  f_\Omega(\SkewVar) = \lclm_{\alpha\in\Omega} \{\SkewVar-\alpha\}.
\end{align}
where $\lclm$ is defined as follows.
\begin{definition} \label{def:lclm}
The least common left multiple (lclm) of $g_i\in \SkewPolys$, denoted by $\lclm_i\{g_i\}$, is the unique monic polynomial $h\in \SkewPolys$ s.t.~$g_i|_r h, \forall i$.
\end{definition}
The \emph{polynomial independence} of a set is defined via its minimal \newinline{polynomial}.
\begin{definition}[P-independent Set {\cite[Def.~2.6]{FnTsurvey-Umberto}}]
  \label{def:PindSet}
  A set $\Omega\subseteq \Fqm$ is P-independent in $\SkewPolys$ if $\deg (f_{\Omega}) = |\Omega|$.
\end{definition}

\begin{proposition}[{\cite[Theorem 8]{lam1986general}}]\label{prop:PindRankVand}
For a P-independent set $\PindSet$,
$\deg(f_{\PindSet})=|\PindSet|=\rank(\bV^{\aut}(\PindSet))$.
\end{proposition}
\begin{proposition}[{\cite[Corollary 2.8]{FnTsurvey-Umberto}}]\label{prop:subsetPind}
  Any subset of a P-independent set is P-independent.
\end{proposition}
\begin{lemma}\label{lem:noMoreZero}
Given a P-independent set $\Omega$, for any subset $\rootSet\subset \Omega$, let $f_{\rootSet}(x)\in\SkewPolys$ be the minimal polynomial of $\rootSet$. Then, for any element $\alpha\in\Omega\setminus\rootSet$, $f_{\rootSet}(\alpha)\neq 0$.
\end{lemma}
\begin{IEEEproof}
Assume $f_{\rootSet}(\alpha)=0$, then the minimal polynomial $f_{\rootSet\cup\{\alpha\}}=f_{\rootSet}$ and $\deg(f_{\rootSet\cup\{\alpha\}}) = |\rootSet|< |\rootSet\cup\{\alpha\}|$, which contradicts to that $\rootSet\cup\{\alpha\}\subseteq \Omega$ is P-independent.
\end{IEEEproof}

\subsection{Linearized Reed-Solomon Codes}
The definition of LRS codes adopted in this paper follows from the \emph{generalized skew evaluations codes} \cite[Section III]{liu2015construction} with particular choices of the evaluation points and column multipliers.

\begin{definition}[Linearized Reed-Solomon (LRS) Codes]
\label{def:LRS}
Let $\ell\leq q-1$ and $(n_1,\dots, n_{\ell})$ be an ordered partition of $n$ with $n_i\leq m$ for all $i=1,\dots,\ell$.
Let $a_1,\dots, a_\ell\in\Fqm$ be from distinct $\Frobaut$-conjugacy classes  of $\Fqm$, and called \emph{block representatives}. 
Let
$$\colMulVec=(\beta_{1,1},\dots,\beta_{1,n_1}\ \vdots\ \cdots\ \vdots \ \beta_{\ell,1},\dots,\beta_{\ell,n_\ell})\in\Fqm^n$$
be a vector of \emph{column multipliers}, where $\betal{1},\dots,\betal{n_l}$ are linearly independent over $\Fq, \forall l\in[\ell]$.

Let the set of \emph{code locators} be
\begin{align}\label{eq:LRSlocators}
  \locSet 
  =&\{a_1\beta_{1,1}^{q-1}, \dots, a_1\beta_{1, n_1}^{q-1}
  \vdots
  \cdots
  \vdots
  a_\ell\beta_{\ell,1}^{q-1}, \dots, a_\ell\beta_{\ell, n_\ell}^{q-1}\}\ .
\end{align}
An $[n,k]_{q^m}$ \emph{linearized Reed-Solomon} code is defined as
\begin{align*}
  \cC_{k}^{\Frobaut}(\locSet, \colMulVec):= \{ \colMulVec \star (f(\alpha))_{\alpha\in\locSet} \  |& \ f(\SkewVar)\in\FrobPolys,\\ &\deg f(\SkewVar)<k\},
\end{align*}
where
the evaluation $f(\alpha) = \sum_{i=0}^{\deg f}f_iN_{i}(\alpha)$ is the remainder evaluation\versionShortLong{}{ as in \cref{lem:evaluation}}.
\end{definition}
The code locator set $\locSet$ of LRS codes \versionShortLong{is P-independent \cite[Theorem 2.11]{FnTsurvey-Umberto}.}{has the following properties.
\begin{proposition}[{\cite[Theorem 4.5]{lam1988vandermonde}}]
  \label{prop:blockPind}
  The code locators in the $\indBlock$-th block, i.e., $\{a_\indBlock\beta_{\indBlock,1}^{q-1},\dots, a_\indBlock\beta_{\indBlock,n_\indBlock}^{q-1}\}$ are P-independent, if and only if $\betal{1},\dots,\betal{n_\indBlock}$ are linearly independent \newinline{over $\Fq$}.
\end{proposition}
\begin{proposition}[{\cite[Theorem 2.11]{FnTsurvey-Umberto}}]
  \label{lem:codeLocatorsPind}
  The union of P-independent sets which are subsets of different conjugacy classes is P-independent. Hence, the code locator set $\cL$ given in \eqref{eq:LRSlocators} is P-independent.
\end{proposition}
}
A generator matrix of the LRS code in \cref{def:LRS} is given by
\begin{align}\label{eq:Gevaluation}
  \GLRS =&
             \left( \quad \GLRSi{1}\quad \newinline{|}\quad  \dots \quad  \newinline{|} \quad \GLRSi{\ell}\quad \right)
           \in\Fqm^{k\times n}
\end{align}
where for each $l\in[\ell]$,
\begingroup
\setlength\arraycolsep{3pt}
\allowdisplaybreaks
\begin{align}
  &\GLRSi{l}
   = \bV^\Frobaut_{k}(\locSet^{(l)})  \cdot \diag(\colMulVec^{(l)}) \nonumber \\
  =& \begin{pmatrix}
    1 & \dots & 1 \\
    N_1(a_l\betal{1}^{q-1})& \dots & N_1(a_l\betal{n_l}^{q-1})\\
    \vdots & \ddots & \vdots \\
    N_{k-1}(a_l\betal{1}^{q-1})& \dots & N_{k-1}(a_l\betal{n_l}^{q-1})
  \end{pmatrix}\cdot
  \begin{pmatrix}
    \betal{1} &&\\
    &\ddots & \\
    && \betal{n_l}
  \end{pmatrix}\nonumber \\
     =&\begin{pmatrix}
       1 &&\\
       &\ddots & \\
       && N_{k-1}(a_l)
     \end{pmatrix}\cdot
           \begin{pmatrix}
             \betal{1} & \betal{2} & \dots & \betal{n_l}\\
             \betal{1}^{q^1} & \betal{2}^{q^1} & \dots & \betal{n_l}^{q^1}\\
               \vdots & \vdots & \ddots & \vdots \\
             \betal{1}^{q^{k-1}} & \betal{2}^{q^{k-1}} & \dots & \betal{n_l}^{q^{k-1}}
    \end{pmatrix} \label{eq:MooreMatrix}
     \ ,
\end{align}
\endgroup
where $\locSet^{(l)}\coloneqq\{a_l\betal{1}^{q-1},\dots,a_l\betal{n_l}^{q-1}\}$ and $\colMulVec^{(l)}\coloneqq (\betal{1},\dots,\betal{n_l})$. Eq.~\eqref{eq:MooreMatrix} holds because for $\Frobaut(a)=a^q$, $N_i(\betalt^{q-1})\cdot \betalt=\left(\betalt^{q-1}\right)^{(q^i-1)/(q-1)}\cdot \betalt= \betalt^{q^i}$.

In \cite[Def.~31]{martinez2018skew}, LRS codes are defined in the notion of linear operator evaluation with respect to the block representatives $\ba=(a_1,\dots,a_\ell)$ and block basis $\bb_i = (\beta_{i,1},\dots,\beta_{i,n_i})$. It was shown in \cite[Theorem 2.18]{FnTsurvey-Umberto} that these two definitions are equivalent.

LRS codes are MSRD codes \cite[Theorem 2.20]{FnTsurvey-Umberto} while they are \emph{maximum distance separable} (MDS) linear codes $\cC\subseteq\Fqm^n$ and \newinline{for small dimensions $k\leq\min\{n_1,\dots,n_{\ell}\}$,} the punctured codes $\cC_i\subseteq \Fqm^{n_i}$ at any block $i=1,\dots,\ell$ are \emph{maximum rank distance} (MRD) codes \cite[Section III.C]{liu2015construction}.


\section{LRS Codes with Support Constraints}
\label{sec:field_size}

In this section we show that \eqref{eq:GM-MDScondition} is also a necessary and sufficient condition that a matrix $\bG$ fulfilling \eqref{eq:zeroConstraints} generates an MSRD code.

Since the sum-rank weight is at most the Hamming weight for any vector in $\Fqm^n$, an MSRD code is necessarily an MDS code.
Therefore, \eqref{eq:GM-MDScondition} is also a necessary condition for $\G$ to generate an MSRD code.

Now we proceed to show the sufficiency of \eqref{eq:GM-MDScondition} for MSRD codes, in particular, LRS codes.
Note that for any $\Omega = \{i\}$, we have $|\zeroSet_i|\leq k-1$. One can add elements from $[n]$ to each $\zeroSet_i$ until $|\zeroSet_i|$ reaches $k-1$ while preserving \eqref{eq:GM-MDScondition} \cite[Corollary 3]{yildiz2019gabidulin}. This operation will only put more zero constraints on $\G$ but not remove any. This means that the code we design under the new $\zeroSet_i$'s of size $k-1$ will also satisfy the original constraints. Therefore, without loss of generality, along with \eqref{eq:GM-MDScondition}, we will further assume that
\begin{align}
  \label{eq:maxRowZeros}
    |\zeroSet_i|= k-1,\ \forall i\in[k]\ .
\end{align}

Let $\GLRS$ be a generator matrix of an LRS code as given in \eqref{eq:Gevaluation}. Given the following matrix
\begin{align}
  \G = \bT\cdot \GLRS \label{eq:defTmat},
\end{align}
if $\bT\in\Fqm^{k\times k}$ has full rank, then $\G$ is another generator matrix of the same LRS code generated by $\GLRS$.
Recall that $a_1,\dots, a_\ell\in\Fqm$ are the block representatives, $\beta_{1,1},\dots,\beta_{1,n_1}, \dots, \beta_{\ell,1},\dots,\beta_{\ell,n_\ell}\in\Fqm$ are the column multipliers, and $\locSet = \{\alphas{n}\}$ is the code locator set as defined in \cref{def:LRS}, where $\alpha_j = a_l\beta^{q-1}_{l, t}$ for some $l\in[\ell]$ and $t\in[n_l]$, $\forall j\in [n]$.
Let $n_0=0$. Define the following bijective map between the indices, $\indMap: \bbN\times \bbN \to \bbN$,
\begin{equation}
  \label{eq:indMap}
    (l,t) \mapsto j=t+\sum_{r=0}^{l-1}n_r,
\end{equation}
such that $\alpha_j = a_l\beta^{q-1}_{t}$.
The inverse map $\indMap^{-1}: \bbN \to \bbN\times \bbN$ is
$j \mapsto (l, t)$,
where $l=\max\{i ~|~ \sum_{r=1}^i n_r\leq j\}$ and $t=j-\sum_{r=0}^{l-1}n_r$.

For all $i\in[k]$, define the skew polynomials
\begin{align}\label{eq:skewPolyEachRow}
  \rowpoly_i(\SkewVar) \coloneqq \sum_{j=0}^{k-1} T_{i,j+1}\SkewVar^{j}\ \in\FrobPolys\ ,
\end{align}
where $T_{i,j+1}$ is the entry at $i$-th ($1\leq i\leq k$) row, $\newinline{(j+1)}$-th ($1\leq j+1\leq k$) column in $\bT$.
The entries of $\G$ will be 
$G_{ij}=\rowpoly_i(a_l\betalt^{q-1})\betalt,i\in[k],j=\indMap(l,t)\in[n]$. Then, the zero constraints in \eqref{eq:zeroConstraints} become root constraints on $\rowpoly_i$'s:
\begin{align}\label{eq:rootConstraints}
   \rowpoly_i(a_l\betalt^{q-1})=0, \quad\forall i \in[k],\ \forall j=\indMap(l,t)\in \zeroSet_i\ .
\end{align}

For brevity, we denote by
\begin{align}
  \label{eq:rootSet}
  \rootSet_i \coloneqq\{a_l\betalt^{q-1}\ |\ \indMap(l,t)\in \zeroSet_i\}
\end{align}
corresponding to the zero set $\zeroSet_i$.
\newinline{Let $\rowpoly_{\rootSet_i}(\SkewVar)= \lclm_{\alpha\in\rootSet_i} \{(\SkewVar-\alpha)\}$ be the minimal polynomial of $\rootSet_i$. It follows from \eqref{eq:rootConstraints} that $\rowpoly_i(\alpha)=0,\ \forall \alpha \in \rootSet_i$ and hence, $\rowpoly_{\rootSet_i}(\SkewVar)|_r\rowpoly_{i}(\SkewVar)$. Note that $\deg \rowpoly_{\rootSet_i}(\SkewVar)=\deg \rowpoly_i(\SkewVar)=k-1$, therefore, $\rowpoly_i(\SkewVar)=c\cdot\rowpoly_{\rootSet_i}(\SkewVar)$ for some $c\in\Fqm\setminus\{0\}$. For simplicity, let $c=1$, i.e.,}
\begin{align}\label{eq:lclmMinPoly}
  \rowpoly_i(\SkewVar) = \rowpoly_{\rootSet_i}(\SkewVar)\ .
\end{align}
Since $\cL$ and any subset $\rootSet_i\subset \cL$ are all P-independent, it follows from \cref{lem:noMoreZero} that $\rowpoly_i(\alpha)\neq 0$, for all $\alpha\in\cL\setminus\rootSet_i$. Hence, there is no other zero in $\G$ than the required positions in $\zeroSet_i$'s.
Moreover, by the assumption in \eqref{eq:maxRowZeros}, $|\rootSet_i|=|\zeroSet_i|= k-1$, and $\deg \rowpoly_i (\SkewVar)=k-1, \forall i\in[k]$.
Hence the coefficients of $\rowpoly_i(\SkewVar)$ in \eqref{eq:skewPolyEachRow} are uniquely determined (up to scaling) in terms of $a_1\beta_{1,1}^{q-1}, \dots,a_\ell\beta_{\ell,n_\ell}^{q-1}$.
In the following, we assume $a_1,\dots,a_\ell$ are fixed, nonzero, and from distinct $\Frobaut$-conjugacy classes. We see $\betalt$'s as variables of the following commutative multivariate polynomial ring
\begin{align}
\label{eq:multiVarRing}
    R_{\numMulPolyVar}\coloneq
    &\Fqm[\beta_{1,1},\dots, \beta_{\ell,n_\ell}],
\end{align}
and the coefficients $T_{i,j+1}$ of $\rowpoly_i(\SkewVar)$ can be seen as polynomials in $R_{\numMulPolyVar}$.
Then the problem of finding $\betalt$'s such that $\bG$ generates the same LRS code as $\GLRS$ becomes finding $\betalt$'s such that
\begin{align}\label{eq:defPl}
P(\beta_{1,1},\dots, \beta_{\ell,n_\ell})\coloneq & P_\bT(\beta_{1,1},\dots, \beta_{\ell,n_\ell}) \nonumber \\
&\cdot \prod_{l=1}^{\ell} P_{\bM_l}(\beta_{l,1},\dots, \beta_{l,n_l})
\neq 0
\end{align}
where $P_\bT$ is the determinant of $\bT$, whose entries are determined by the minimal polynomials $f_i$'s, and
\begin{align*}
  P_{\bM_l}
  \coloneqq & \det
          \begin{pmatrix}
             \betal{1} & \betal{2} & \dots & \betal{n_l}\\
             \betal{1}^{q^1} & \betal{2}^{q^1} & \dots & \betal{n_l}^{q^1}\\
              \vdots & \vdots & \ddots & \vdots \\
             \betal{1}^{q^{n_l-1}} & \betal{2}^{q^{n_l-1}} & \dots & \betal{n_l}^{q^{n_l-1}}
          \end{pmatrix}\ . 
\end{align*}
Since the coefficient of the monomial $\prod_{i=1}^{n_l} \beta_{l,i}^{q^{i-1}}$ in  $P_{\bM_l}$ is $1$, $P_{\bM_l}$ is a nonzero polynomial in $R_{\numMulPolyVar}$.
With \cref{claim} below, we can conclude that $P(\beta_{1,1},\dots, \beta_{\ell,n_\ell})$ is a nonzero polynomial in $R_{\numMulPolyVar}$.

\begin{claim}\label{claim}
If the condition in \eqref{eq:GM-MDScondition} is satisfied, then $P_{\bT}$ is a nonzero polynomial in $R_\numMulPolyVar$.
\end{claim}
Now we proceed to present the result on the field size by assuming \cref{claim} is true.
A more general version (\cref{thm:sufficientCond}) of the claim is given in \cref{sec:ind_proof}.

For a fixed $l\in[\ell],t\in[n_\ell]$, the degree in $\betalt$ of $P_{\bM_l}$ is $\deg_{\betalt}P_{\bM_l} = q^{n_l-1}$ \cite[Lemma 3.51]{lidl1997finite}. Moreover, $\deg_{\betalt}P_\bT \leq (k-1)(q-1)\cdot q^{k-2}$, which can be shown by extending the analysis of linearized polynomials for Gabidulin codes in \cite[Section II.F]{yildiz2019gabidulin} to skew polynomials. The details of this extension are provided in \versionShortLong{\cite[Appendix B]{full_version}}{\cref{apendix:fieldSizeSkewPoly}}.
Then, the degree of $P(\beta_{1,1},\dots, \beta_{\ell,n_\ell})$ in \eqref{eq:defPl} as a polynomial in $\betalt$ is
\begin{align*}
  \deg_{\betalt}P &\leq (k-1)(q-1)\cdot q^{k-2}+q^{n_l-1}\ .
\end{align*}

\begin{theorem}\label{thm:fieldSizeFromGab}
  Let $\ell,n_l$ be positive integers 
  and $n \coloneqq \sum_{l=1}^\ell n_l$.
  Let $\zeroSet_1,\dots,\zeroSet_k\subset [n]$ fulfill \eqref{eq:GM-MDScondition} for any nonempty $\Omega\subseteq [k]$. Then for any prime power $q\geq\ell+1$ and integer $m\geq \max_{l\in[\ell]}\{k-1+\log_qk, n_l\}$, there exists an $[n,k]_{q^m}$ linearized Reed-Solomon code with $\ell$ blocks, and each block of length $n_l$, $l\in[\ell]$ with a generator matrix $\bG\in\Fqm^{k\times n}$ fulfilling the support constraints in \eqref{eq:zeroConstraints}.
\end{theorem}
\begin{IEEEproof}
  \newinline{Recall that $a_1,\dots,a_{\ell}$ are fixed nonzero elements from distinct $\Frobaut$-conjugacy classes.}
  \cref{claim} has shown that $P(\beta_{1,1},\dots, \beta_{\ell,n_\ell})$ is a nonzero polynomial.
  By the Combinatorial Nullstellensatz \cite[Theorem 1.2]{alon1999combinatorial}\versionShortLong{,}{(see \cref{lem:nullstellensatz}),}
  there exist $\evapt{\beta}_{1,1}, \dots, \evapt{\beta}_{\ell, n_\ell}$ in $\Fqm$ such that  $$P(\evapt{\beta}_{1,1}, \dots, \evapt{\beta}_{\ell, n_\ell})\neq 0$$
  if
  \begin{align}
    q^m &>\max_{l\in[\ell], t\in[n_l]}\{\deg_{\betalt}P\}  \nonumber\\
        &=\max_{l\in[\ell]}\{ (k-1)(q-1)\cdot q^{k-2}+q^{n_l-1}\}\ .\label{eq:condition_on_qm_2}
  \end{align}
  If $m\geq \max_{l\in[\ell]}\{k-1+\log_qk\ ,\ n_l\}$, we have
  \begin{align*}
    q^m =& (q-1)q^{m-1}+q^{m-1}\\
    \geq& \max_{l\in[\ell]}\{k(q-1)\cdot q^{k-2} + q^{n_l-1}\} >\eqref{eq:condition_on_qm_2}\ .
  \end{align*}
  To have $a_1,\dots,a_\ell$ from different nontrivial $\Frobaut$-conjugacy class of $\Fqm$, by \versionShortLong{the structure of $\Frobaut$-conjugacy classes \cite[Theorem 2.12]{FnTsurvey-Umberto}}{\cref{thm:conjugacyClasses}}, we require $q-1\geq \ell$.
\end{IEEEproof}
\begin{remark}
Consider the extreme cases:
\begin{enumerate}
\item
For $\ell=1$, the sum-rank metric is the rank metric and LRS codes are Gabidulin codes.
\item
For $\ell=n$ and $n_l=1,\forall l\in[\ell]$, the sum-rank metric is the Hamming metric. In addition, with $\theta=\mathrm{Id}$, LRS codes are generalized RS codes with distinct nonzero $a_1,\dots,a_\ell$ as code locators and nonzero $\betalt$'s as column multipliers (see \cite[Theorem 2.17]{FnTsurvey-Umberto}, \cite[Table II]{martinez2019reliable}).
\end{enumerate}
For the first case, our result on the field size in \cref{thm:fieldSizeFromGab} coincides with \cite[Theorem 1]{yildiz2019gabidulin}. For the second case, by adapting the setup in \eqref{eq:multiVarRing}-\eqref{eq:defPl} to $\theta=\mathrm{Id}$, and the proof in \versionShortLong{\cite[Appendix B]{full_version}}{ \cref{apendix:fieldSizeSkewPoly}} with the usual evaluation of commutative polynomials, one can obtain the same results as in \cite[Theorem~2]{yildiz2018optimum}.
\end{remark}

  If the necessary and sufficient condition on $\zeroSet_1,\dots,\zeroSet_k$ in \eqref{eq:GM-MDScondition} is not satisfied, we cannot obtain an MSRD code fulfilling the zero constraints. The following result shows that the largest possible sum-rank distance can be achieved given any zero constraints. In fact, the largest sum-rank distance can be achieved by subcodes of LRS codes with large enough field sizes. This result is an analogue to those for MDS codes \cite{yildiz2018optimum} and MSRD codes \cite{yildiz2019gabidulin}.
  In \cite[Theorem 1]{yildiz2018optimum}, the following upper bound on the Hamming distance of a code with support-constrained generator matrix is given
  \begin{align*}
    \dH \leq n-\widetilde{k}+1
  \end{align*}
  where
  \begin{align}
    \label{eq:kSuperCode}
    \widetilde{k} \coloneqq \max_{\varnothing \neq \Omega\subseteq [k]}  \left|\bigcap_{i\in\Omega} \zeroSet_i\right|+|\Omega|.
  \end{align}
  Note that $\widetilde{k}>k$ if the condition on $\zeroSet_1,\dots,\zeroSet_k$ in \eqref{eq:GM-MDScondition} is not satisfied.
  For any ordered partition $\bn_\ell =(n_1, \dots,n_\ell)$ of $n$, according to \cref{lem:sum-rank-Hamming}, we have
  \begin{align}\label{eq:distanceSubcode}
    \dSR{\bn_{\ell}} \leq n-\widetilde{k}+1\ .
  \end{align}
  \begin{theorem}\label{thm:subcode}
    Given zero constraints $\zeroSet_1,\dots,\zeroSet_k\subseteq[n]$, for any prime power $q\geq\ell+1$ and integer $m\geq \max_{l\in[\ell]}\{\widetilde{k}-1+\log_q\widetilde{k}, n_l\}$, there exists a subcode of an $[n,\widetilde{k}]_{q^m}$ linearized Reed-Solomon code with $\ell$ blocks, and each block of length $n_l$, $l\in[\ell]$ such that its generator matrix satisfies \eqref{eq:zeroConstraints}.
  \end{theorem}
  \begin{IEEEproof}
    Let $\zeroSet_{k+1}=\cdots = \zeroSet_{\widetilde{k}} = \varnothing$. Then for any nonempty $\Omega\subseteq [\widetilde{k}]$, we have
    \begin{align*}
      \left|\bigcap_{i\in\Omega}\zeroSet_i\right|+|\Omega|\leq \widetilde{k}\ .
    \end{align*}
    Then, by \cref{thm:fieldSizeFromGab}, there exists an LRS code of dimension $\widetilde{k}$ with a $\widetilde{k}\times n$ generator matrix $\bG$ having zeros specified by $\zeroSet_1, \dots,\zeroSet_{\widetilde{k}}$. Since it is an MSRD code, its sum-rank distance is $n-\widetilde{k}+1$. The first $k$ rows of $\bG$ will generate a subcode whose sum-rank distance is as good as the LRS codes, i.e., $\dSR{\bn_{\ell}}\geq n-\widetilde{k}+1$, where $\bn_{\ell}=(n_1,\dots, n_\ell)$. Hence, the subcode achieves the largest possible distance given in \eqref{eq:distanceSubcode}.
  \end{IEEEproof}


\section{Application to Distributed Multi-Source Network Coding}
\label{sec:app_network_cod}

  Consider a \emph{distributed multi-source network} as illustrated in \cref{fig:RLNmodel}. The receiver at the sink intends to obtain all the messages in a set $\cM$ by downloading through an $\Fq$-linear network from multiple source nodes, where each of them has access to only a few messages in $\cM$. This source node access is assumed to have unlimited link capacity (e.g., the source nodes store the subset of $\cM$ locally) \newinline{and marked by the dashed line in \cref{fig:RLNmodel}. The capacity\footnote{The number of symbols in $\Fqm$ or the number of vectors in $\Fq^m$ can be sent in one time slot.} of the link between each source node $S_{\cJ_i}$ and the $\Fq$-linear network is $n_{\cJ_i}$, the capacity of a link between any pair of node in the $\Fq$-linear network is $1$, and the capacity of the link between the $\Fq$-linear network and the sink is $N$.}
The topology of the $\Fq$-linear network is not known to the source nodes nor to the sink; therefore, it is a \emph{noncoherent} communication scenario.
This model can find its applications in data sharing platforms, sensor networks, satellite communication networks and MIMO (massive input and output) attenna communication systems, etc.
  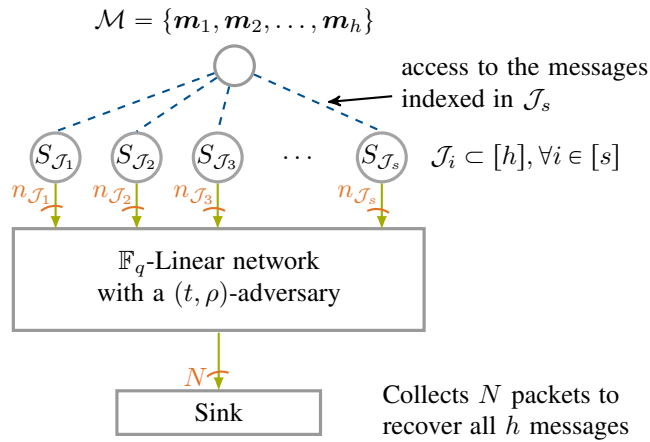
\begin{figure}[h]
    \centering
    \def\x{0.55}

\begin{tikzpicture}
  [font=\normalsize,>=stealth',
  mycircle/.style={circle, draw=TUMGray, very thick, text width=.1em, minimum height=1.5em, text centered},
  myrectangle/.style={rectangle, draw=TUMGray, very thick, text width=\linewidth, minimum height=1.5em, text centered},
  mycircle_small/.style={circle,draw=TUMGray!90,very thick, inner sep=0,minimum size=1em,text centered},
  mylink/.style={color=TUMBlueDark, thick, dashed},
  mylink_net/.style={color=TUMGreen, thick},
  myarrow/.style={->, color=black, thick},
  myarc/.style={color=TUMOrange, thick},
  ]
  \coordinate (Mset) at (0*\x,4*\x);
  {\node[mycircle,label=above:{$\cM=\{\Msym_1,\Msym_2,\dots,\Msym_h\}$}] (Mset) {};}
  \node[mycircle_small,below left = \x*40pt and \x*100pt of Mset] (M0) {$S_{\cJ_1}$};
  \node[mycircle_small,right = \x*20pt of M0] (M1) {$S_{\cJ_2}$};
  \node[mycircle_small,right = \x*20pt of M1] (M2) {$S_{\cJ_3}$};
  \node[draw=none,right = \x*20pt of M2] (M3) {$\dots$};
  \node[mycircle_small,right = \x*20pt of M3] (M4) {$S_{\cJ_s}$};

  \node[draw=none,right= 5pt of M4] (cJ) {$\cJ_i\subset [h], \forall i \in[s]$};
  \node[below left= \x*30pt and \x*10pt of M0] (R0) {};
  \node[below right= \x*100pt and \x*10pt of M4] (R4) {};
  \draw[myrectangle] (R0) rectangle (R4) node[pos=.5, text width=5cm] (RLN) {$\Fq$-Linear network with a $(t,\rho)$-adversary};

  \path[] (Mset) edge[mylink] (M0.north)
  edge[mylink] node [midway] (link1) {} (M1.north)
  edge[mylink] (M2.north)
  edge[mylink] node [midway] (links) {} (M4.north);

  \node[draw=none,above= \x*10pt of cJ, text width=\x*6cm] (T1) {access to the messages indexed in $\cJ_s$};
  \path[] (T1) edge[myarrow] (links);

  \path[-latex] (M0.south) edge[mylink_net] node (S1N) {} ($(M0.south)-(0, \x*32pt)$);
  \path[-latex] (M1.south) edge[mylink_net] node (S2N) {} ($(M1.south)-(0, \x*32pt)$);
  \path[-latex] (M2.south) edge[mylink_net] node (S3N) {} ($(M2.south)-(0, \x*32pt)$);
  \path[-latex] (M4.south) edge[mylink_net] node (SsN) {} ($(M4.south)-(0, \x*32pt)$);

  \draw [myarc] ($(S1N)-(\x*10pt,\x*3pt)$) arc (130:70:\x*15pt);
  \node[draw=none, myarc] at ($(M0.south)-(\x*16pt, \x*8pt)$) (c1) {$n_{\cJ_1}$};
  \draw [myarc] ($(S2N)-(\x*10pt,\x*3pt)$) arc (125:65:\x*15pt);
  \node[draw=none, myarc] at ($(M1.south)-(\x*16pt, \x*8pt)$) (c2) {$n_{\cJ_2}$};
  \draw [myarc] ($(S3N)-(\x*10pt,\x*3pt)$) arc (120:60:\x*15pt);
  \node[draw=none, myarc] at ($(M2.south)-(\x*16pt, \x*8pt)$) (c3) {$n_{\cJ_3}$};
  \draw [myarc] ($(SsN)-(\x*10pt,\x*3pt)$) arc (115:55:\x*15pt);
  \node[draw=none, myarc] at ($(M4.south)-(\x*16pt, \x*8pt)$) (cJ) {$n_{\cJ_s}$};

  \node[draw=none] at ($(RLN.south)-(-\x*70pt,\x*50pt)$) (sink1) {};
  \node[draw=none] at ($(RLN.south)-(\x*70pt,\x*80pt)$) (sink2) {};
  \draw[myrectangle] (sink1) rectangle (sink2) node[pos=.5] (sink) {Sink};
  \node[right=-\x*13cm of sink,text width=4cm] (recover) {Collects $N$ packets to recover all $h$ messages};

  \path[-latex] ($(RLN.south)-(0, \x*10pt)$) edge[mylink_net] (sink.north);
  \draw [myarc] ($(sink.north)-(\x*8pt,-\x*10pt)$) arc (130:70:\x*15pt);
  \node[draw=none, myarc] at ($(sink.north)-(\x*16pt, -\x*10pt)$) (sinkN) {$N$};




\end{tikzpicture}

    \caption{Illustration of the distributed multi-source network model.}
    \label{fig:RLNmodel}
  \end{figure}

  The set $\cM$ contains $h$ messages. The message $\bm_j, j\in[h]$ is composed of $r_j$ symbols over $\Fqm$, i.e., $\bm_j\in\Fqm^{r_j}$.
  The source node $S_{\cJ_i}, i\in[s]$ has access only to the messages indexed in $\cJ_i$, e.g., if $\cJ_2=\{3,6\}$, then $S_{\cJ_2}$ only has the access to the messages $\bm_3$ and $\bm_6$.
    Let $\cS=\{\cJ_1,\dots, \cJ_s\}$. For any $\cJ\in\cS$, the source node $S_{\cJ}$ encodes the messages $\bm_j, j\in\cJ$, into $n_{\cJ}$ symbols over $\Fqm$, denoted by $\bc_{\cJ}\in\Fqm^{n_\cJ}$. It then extends them to their matrix representation over $\Fq$, denoted by $\bC_{\cJ}\in\Fq^{m\times n_{\cJ}}$, and then generates $\bX_{\cJ}= (\0\ \cdots \ \bI_{n_{\cJ}}\ \0 \ \cdots\ \bC_{\cJ}^{\top})\in\Fq^{n_{\cJ}\times (n+m)}$, where $n=\sum_{\cJ\in\cS}n_{\cJ}$. We call each row of $\bX_{\cJ}$ a \emph{packet}.
  Denote $\bX=
  \begin{pmatrix}
    \bX_{\cJ_1}\\
    \vdots\\
    \bX_{\cJ_s}
  \end{pmatrix}\in\Fq^{n\times (n+m)}
  $, where the rows are the packets transmitted by all the source nodes into the $\Fq$-linear network.
  The task is to design $n_\cJ$ for all $\cJ\in\cS$ such that the sink can recover all the messages $\bm_i$. The goal of the design is that the total number of packets $n$ is minimized.
  A concrete example can be found in \cref{sec:ex-dist-LRS}.

  In the $\Fq$-linear network, whenever there is a transmission opportunity, a relay node in the network produces and sends an arbitrary $\Fq$-linear combination of all the incoming packets they have received.
  Suppose that there are at most $t$ \emph{malicious} nodes that inject erroneous packets and at most $\rho$ \emph{frozen} nodes that do not send any packet, which we refer as a \emph{$(t,\rho)$-adversary}.
 The sink collects $N\geq n-\rho$ packets, which are represented by the rows of $\bY\in\Fq^{N\times (n+m)}$.
 The transmitted packets (rows of $\bX$) and the received packets (rows of $\bY$) can be related via the following network equation:
  \begin{align}
    \label{eq:net-equation}
    \bY=\bA\bX+\bE
  \end{align}
  where $\bA\in\Fq^{N\times n}$ is the \emph{transfer} matrix of the network and the difference between the number of columns and its row-rank is at most $\rho$. In other words, $n-\rank(\bA)\leq \rho$.
  $\bE\in\Fq^{N\times M}$ is an error matrix of $\rank(\bE)\leq t$.
  Note that the matrices $\bA$ and $\bE$ are not known to any of the source nodes or the sink since we consider a noncoherent communication scenario.

    The capacity region of a multi-source network with $h$ messages is a set $\{(r_1,\dots, r_h)\}\subseteq\bbN^h$ such that the receiver at the sink can recover \newinline{every} message $\bm_j\in\Fqm^{r_j}, j\in[h]$.
    The capacity region of a multi-source network against a $(t,\rho)$-adversary has been given in \cite[Theorem 2]{dikaliotis2011multiple} (for $\rho=0$) and \cite[Corollary 66]{ravagnani2018adversarial}.
    To present the result, we require the following definitions of \emph{min-cut}.
    \begin{definition}[Min-cut between a set of nodes and another node]
      For a directed graph $\cG(\cV,\cE)$ composed of a set of nodes $\cV$ and a set of edges $\cE$,
      a \emph{cut} between a set of nodes $\cV'\subset\cV$ and another node $t\in\cV\setminus\cV'$ is a subset $\cE_{\cV',t}\subseteq\cE$ such that after removing the edges in $\cE_{\cV',t}$, there is no path from any of the nodes in $\cV'$ to $t$. \newinline{The \emph{capacility of a cut} is the sum of the capacility of each edge in the cut.} The \emph{min-cut} between $\cV'$ and $t$ is the smallest cardinality of a cut between $\cV'$ and $t$.
    \end{definition}
    \begin{definition}[Min-cut between a subset of messages and the sink]
        Consider the multi-source network with $h$ messages, as above. Given a subset of messages, $\cJ'\subseteq[h]$, consider the set of source nodes $\cV'$ that contain messages in $\cJ'$, namely,
        \[
        \cV' = \{ S_\cJ \in \cS ~:~ \cJ\cap\cJ'\neq\emptyset\}.
        \]
        We define the \emph{min-cut between $\cJ'$ and the sink} as the min-cut between $\cV'$ and the sink, and denote it by $\mincut_{\cJ'}$.
    \end{definition}
    \begin{theorem}[\hspace{1pt}\cite{dikaliotis2011multiple, ravagnani2018adversarial}]
      \label{thm:capacity-distributed-network}
      Consider a multi-source network with $h$ messages. For any $(r_1,\dots, r_h)\in\bbN^h$ in the capacity region against a $(t,\rho)$-adversary, we have
    \begin{align}
      \label{eq:multi-source-capacity}
      \forall \varnothing \neq \cJ'\subseteq [h],\ \sum_{i\in\cJ'} r_i\leq \mincut_{\cJ'} - 2t-\rho\ ,
    \end{align}
    where $\mincut_{\cJ'}$ is the min-cut between the set ${\cJ'}$ of messages and the sink.
    \end{theorem}
  In addition to the general setting, we further assume the following setup of the noncoherent network:
  \begin{itemize}
  \item
    The communication capacity of the non-coherent linear network is large enough so that the min-cut $\mincut_{\cJ'}$ for all $\cJ'\subseteq[h]$ is determined by the number of encode symbols $n_{\cJ}$ sent by the source node $S_{\cJ}$ for all $\cJ\in\cS$. I.e.,
      \begin{align*}
        \mincut_{\cJ'}= n - \sum_{\substack{\cJ\in\cS\\\cJ\subseteq [h]\setminus\cJ'}}n_{\cJ}\ .
      \end{align*}
      Note that the term $\sum_{\substack{\cJ\in\cS\\\cJ\subseteq [h]\setminus\cJ'}} n_{\cJ}$ is the total number of encoded symbols that do not contain any information about the messages in $\cJ'$.
  \item Although the encoding is distributed (since each source node may access only a few messages), there is a centralized coordination unit designing the overall code, and the sink knows the distributed code.
  \end{itemize}

  \subsection{Sum-Rank Weight of Error and Erasure with Constrained Rank Weight}
  \label{sec:sum-rank-weight-rank-error}
  In the following, we intend to use LRS codes for the distributed multi-source linear network model.
  Note that the errors and erasures in the $(t,\rho)$-adversarial model are measured in the rank metric. However, LRS codes are used to deal with errors and erasures in the sum-rank metric.
  Hence, we first look into the sum-rank deficiency of the network transfer matrix $\bA\in\Fq^{ N\times n}$ and the sum-rank weight of the error matrix $\bE\in\Fq^{N\times M}$.

  Let $\ell\in\bbN$ and $\bn_\ell=(n_1,\dots, n_\ell)$ be an ordered partition of $n$.
  By \cref{lem:rank-leq-sumrank-mat}, we have
  \begin{align}
    \label{eq:sumrank-weight-erasure}
    \wtSR{\bn_\ell}(\bA)\geq \rank(\bA)\geq n-\rho\ .
  \end{align}
  Hence the sum-rank weight of the erasure induced by the rank-deficient $\bA$ is at most $\rho$.

  For the error $\bE$, consider an ordered partition $\bN_{\ell}=(N_1,\dots, N_\ell)$ of $N$ such that

  \pgfkeys{tikz/mymatrixenv/.style={decoration={brace},every left delimiter/.style={xshift=5pt, yshift=-1pt},every right delimiter/.style={xshift=-5pt, yshift=-1pt}}}

\pgfkeys{tikz/mymatrix/.style={matrix of math nodes,nodes in empty cells,left delimiter={(},right delimiter={)},inner sep=2pt,outer sep=2pt,column sep=4pt,row sep=4pt,nodes={minimum width=3pt,minimum height=2pt,anchor=center,inner sep=0pt,outer sep=0pt}}}

\pgfkeys{tikz/mymatrixbrace/.style={decorate,thick}}

\tikzset{style green/.style={
    set fill color=TUMGreenDark!80!lime!20,fill opacity=0.3,
    set border color=TUMGreenDark!60!lime!40,draw opacity=1.0,
  },
  style cyan/.style={
    set fill color=cyan!90!blue!60, draw opacity=0.4,
    set border color=blue!70!cyan!30,fill opacity=0.1,
  },
  style orange/.style={
    set fill color=TUMOrange!80,fill opacity=0.3,
    set border color=TUMOrange!90,  draw opacity=0.8,
  },
  style brown/.style={
    set fill color=brown!70!orange!40, draw opacity=0.4,
    set border color=brown, fill opacity=0.3,
  },
  style purple/.style={
    set fill color=violet!90!pink!20, draw opacity=0.5,
    set border color=violet, fill opacity=0.3,
  },
  style black/.style={
    set fill color=none, fill opacity=0.3,
    set border color=black, draw opacity=0.5,
  },
  kwad/.style={
    above left offset={-0.07,0.23},
    below right offset={0.07,-0.23},
    #1
  },
  pion/.style={
    above left offset={-0.07,0.2},
    below right offset={0.07,-0.32},
    #1
  },
  border/.style={
    above left offset={-0.03,0.18},
    below right offset={0.03,-0.3},
    #1
  },
  set fill color/.code={\pgfkeysalso{fill=#1}},
  set border color/.style={draw=#1}
}

\[  
  \begin{tikzpicture}[baseline={-0.5ex},mymatrixenv]
    \matrix [mymatrix,inner sep=4pt, row sep=6pt] (A)
    {
      \bA_{1,1} & \bA_{1,2} & \cdots & \bA_{1,\ell} \\
      \bA_{2,1} & \bA_{2,2} & \cdots & \bA_{2,\ell} \\
      \vdots & \vdots & \vdots & \vdots \\
      \bA_{\ell,1} & \bA_{\ell,2} & \cdots & \bA_{\ell,\ell} \\
    };
    \matrix [right= 8pt of A, mymatrix, inner sep=4pt] (X)
    {
      \bX_{1} \\
      \bX_{2} \\
      \vdots \\
      \bX_{\ell} \\
    };
    \node[right =30pt of X] (plus) {$+$};

    \matrix [mymatrix, right = 5pt of plus, row sep=8pt, column sep=6pt] (E)
    {
      \bE_1\\
      \bE_2\\
      \vdots\\
      \bE_{\ell}\\
    };
    \node[left =30pt of A] (eq) {$\bA\bX+\bE =$};
        \mymatrixbracebottom{A}{1}{1}{$n_1$}
        \mymatrixbracebottom{A}{2}{2}{$n_2$}
        \mymatrixbracebottom{A}{4}{4}{$n_{\ell}$}
        \mymatrixbraceleft{A}{1}{1}{$N_1$}
        \mymatrixbraceleft{A}{2}{2}{$N_2$}
        \mymatrixbraceleft{A}{4}{4}{$N_{\ell}$}

        \mymatrixbraceright{X}{1}{1}{$n_1$}
        \mymatrixbraceright{X}{2}{2}{$n_2$}
        \mymatrixbraceright{X}{4}{4}{$n_4$}
        \mymatrixbracebottom{X}{1}{1}{$M$}

        \mymatrixbracebottom{E}{1}{1}{$M$}
        \mymatrixbraceright{E}{1}{1}{$N_1$}
        \mymatrixbraceright{E}{2}{2}{$N_2$}
        \mymatrixbraceright{E}{4}{4}{$N_{\ell}$}

    \end{tikzpicture}
\]

  Given $\rank(\bE)\leq t$, by \cref{lem:rank-leq-sumrank-mat}, we have
  \begin{align}
    \label{eq:sumrank-weight-E}
    \wtSR{\bN_\ell}(\bE) = \sum_{i=1}^{\ell} \rank(\bE_i)
    \leq \sum_{i=1}^{\ell} \rank(\bE) = \ell t\ .
  \end{align}
  This upper bound holds for any arbitrary $\ell$-ordered partition $\bN_{\ell}$ of $N$.
    A lower bound on $\Pr[\wtSR{\bN_\ell}(\bE)=\ell t\ |\ \rank(\bE)=t]$ (i.e., the probability that \eqref{eq:sumrank-weight-E} is tight) for small $t$ ($t\leq N_i, \forall i\in[\ell]$) is given in \cite[Theorem 1]{couvee2023notes}. In particular, if $q\geq \ell+1$, $\Pr[\wtSR{\bN_\ell}(\bE)=\ell t\ |\ \rank(\bE)=t]>1/4$ \cite[Corollary 1]{couvee2023notes}.

  It can been seen from \eqref{eq:sumrank-weight-erasure} and \eqref{eq:sumrank-weight-E} that the network model in \eqref{eq:net-equation} results in an erasure of sum-rank weight at most $\rho$ and an error of sum-rank weight at most $\ell t$.
  It has been shown in \cite[Theorem 1, Eq.(4), Proposition 2]{martinez2019reliable} that a code with sum-rank distance $d$ can guarantee reliable communication against errors of sum-rank weight at most $\ell t$ and erasures with sum-rank weight at most $\rho$ in the noncoherent communication if $d \geq 2\ell t+\rho+1$.
  Therefore, an LRS code with sum-rank distance $d\geq 2\ell t+\rho+1$ can correct any error of rank weight at most $t$ and erasure of rank weight at most $\rho$. 

  \subsection{Example of Distributed LRS codes}
  \label{sec:ex-dist-LRS}
  We first give a toy example to show the usage of LRS codes for a distributed multi-source network. In \cref{sec:distributed-LRS} we provide the general scheme to design the LRS codes for an arbitrary distributed multi-source network.

  Suppose that there are $h=4$ messages in $\cM$. The lengths of messages are $(r_1, r_2, r_3, r_4)=(1,3,2,3)$. There are $4$ source nodes and each can access to only $3$ messages, i.e., $\cJ_1= \{1,2,3\}, \cJ_2= \{1,2,4\}, \cJ_3 = \{1,3,4\}, \cJ_4 = \{2,3,4\}$.
  Suppose there is a $(t=2,\rho=2)$-adversary in the $\Fq$-linear network.

  The number of encoded packets from each source node is $(n_{\cJ_1}, n_{\cJ_2}, n_{\cJ_3}, n_{\cJ_4})= (6,7,2,8)$ (see in \cref{sec:distributed-LRS} how these values are obtained) and $n=\sum_{i=1}^4 n_{\cJ_i} = 23$.
  Let $\bm=(\bm_1,\bm_2,\bm_3,\bm_4)$ be a concatenated vector of all messages. Some entries in a encoding matrix $\bG$ are forced to be $0$, as shown in \cref{fig:eg-support-constrained-LRS}, so that $\bm\cdot \bG$ represents the overall encoding at all source nodes.
  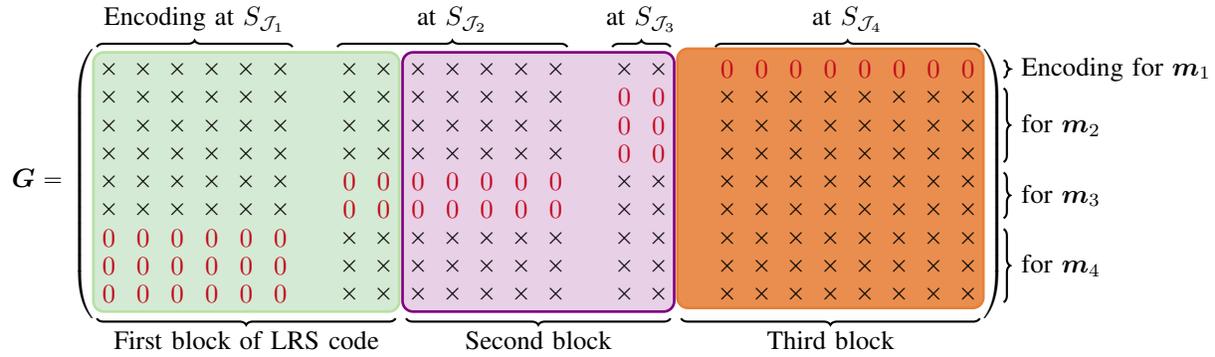
\begin{figure}[h]
    \centering
    \pgfkeys{tikz/mymatrixenv/.style={decoration={brace},every left delimiter/.style={xshift=5pt, yshift=-1pt},every right delimiter/.style={xshift=-5pt, yshift=-1pt}}}

\pgfkeys{tikz/mymatrix/.style={matrix of math nodes,nodes in empty cells,left delimiter={(},right delimiter={)},inner sep=1pt,outer sep=3pt,column sep=4pt,row sep=4pt,nodes={minimum width=9pt,minimum height=6pt,anchor=center,inner sep=0pt,outer sep=0pt}}}

\pgfkeys{tikz/mymatrixbrace/.style={decorate,thick}}

\tikzset{style green/.style={
    set fill color=TUMGreenDark!80!lime!20,fill opacity=0.3,
    set border color=TUMGreenDark!60!lime!40,draw opacity=1.0,
  },
  style cyan/.style={
    set fill color=cyan!90!blue!60, draw opacity=0.4,
    set border color=blue!70!cyan!30,fill opacity=0.1,
  },
  style orange/.style={
    set fill color=TUMOrange!80,fill opacity=0.3,
    set border color=TUMOrange!90,  draw opacity=0.8,
  },
  style brown/.style={
    set fill color=brown!70!orange!40, draw opacity=0.4,
    set border color=brown, fill opacity=0.3,
  },
  style purple/.style={
    set fill color=violet!90!pink!20, draw opacity=0.5,
    set border color=violet, fill opacity=0.3,
  },
  kwad/.style={
    above left offset={-0.07,0.23},
    below right offset={0.07,-0.23},
    #1
  },
  pion/.style={
    above left offset={-0.07,0.2},
    below right offset={0.07,-0.32},
    #1
  },
  poz/.style={
    above left offset={-0.03,0.18},
    below right offset={0.03,-0.3},
    #1
  },set fill color/.code={\pgfkeysalso{fill=#1}},
  set border color/.style={draw=#1}
}

\[
   \bG =
    \begin{tikzpicture}[baseline={-0.5ex},mymatrixenv]
        \matrix [mymatrix,inner sep=2pt] (G)
        {
          \tikzmarkin[kwad=style green]{LRS-first} \times  &  \times &  \times&  \times&  \times&  \times&
          \hphantom{\times} &  \times&  \times& \tikzmarkin[kwad=style purple]{LRS-second} \times &  \times&  \times&  \times & \times&
          \hphantom{\times} &  \times&  \times&
          \tikzmarkin[kwad=style orange]{LRS-third} \hphantom{\times} & \textcolor{TUMRed}{0} & \textcolor{TUMRed}{0} & \textcolor{TUMRed}{0} & \textcolor{TUMRed}{0} & \textcolor{TUMRed}{0} & \textcolor{TUMRed}{0} & \textcolor{TUMRed}{0} & \textcolor{TUMRed}{0}\\
    \times&  \times& \times&  \times& \times&  \times&  \hphantom{\times} &
    \times&  \times& \times&  \times& \times&  \times&  \times & \hphantom{\times} &
    \textcolor{TUMRed}{0} & \textcolor{TUMRed}{0} & \hphantom{\times}&
    \times&  \times& \times&  \times& \times & \times&  \times& \times\\
    \times&  \times& \times&  \times& \times&  \times&  \hphantom{\times} &
    \times&  \times& \times&  \times& \times&  \times&  \times & \hphantom{\times} &
    \textcolor{TUMRed}{0} & \textcolor{TUMRed}{0} & \hphantom{\times}&
    \times&  \times& \times&  \times& \times & \times&  \times& \times\\
    \times&  \times& \times&  \times& \times&  \times&  \hphantom{\times} &
    \times&  \times& \times&  \times& \times&  \times&  \times & \hphantom{\times} &
    \textcolor{TUMRed}{0} & \textcolor{TUMRed}{0} & \hphantom{\times}&
    \times&  \times& \times&  \times& \times & \times&  \times& \times\\
    \times&  \times& \times&  \times& \times&  \times&  \hphantom{\times} &
    \textcolor{TUMRed}{0} & \textcolor{TUMRed}{0} & \textcolor{TUMRed}{0} & \textcolor{TUMRed}{0} & \textcolor{TUMRed}{0} & \textcolor{TUMRed}{0} & \textcolor{TUMRed}{0} &  \hphantom{\times} &
    \times&  \times& \hphantom{\times} &
    \times&  \times& \times&  \times& \times&  \times&  \times&  \times\\
    \times&  \times& \times&  \times& \times&  \times&  \hphantom{\times} &
    \textcolor{TUMRed}{0} & \textcolor{TUMRed}{0} & \textcolor{TUMRed}{0} & \textcolor{TUMRed}{0} & \textcolor{TUMRed}{0} & \textcolor{TUMRed}{0} & \textcolor{TUMRed}{0} &  \hphantom{\times} &
    \times&  \times& \hphantom{\times} &
    \times&  \times& \times&  \times& \times&  \times&  \times&  \times\\
    \textcolor{TUMRed}{0}  & \textcolor{TUMRed}{0} & \textcolor{TUMRed}{0} & \textcolor{TUMRed}{0} & \textcolor{TUMRed}{0} & \textcolor{TUMRed}{0} & \hphantom{\times} &
    \times & \times & \times & \times & \times & \times & \times &\hphantom{\times} &
    \times & \times &\hphantom{\times} &
    \times  & \times & \times & \times & \times & \times & \times & \times \\
    \textcolor{TUMRed}{0}  & \textcolor{TUMRed}{0} & \textcolor{TUMRed}{0} & \textcolor{TUMRed}{0} & \textcolor{TUMRed}{0} & \textcolor{TUMRed}{0} & \hphantom{\times} &
    \times & \times & \times & \times & \times & \times & \times &\hphantom{\times} &
    \times & \times &\hphantom{\times} &
    \times  & \times & \times & \times & \times & \times & \times & \times \\
    \textcolor{TUMRed}{0}  & \textcolor{TUMRed}{0} & \textcolor{TUMRed}{0} & \textcolor{TUMRed}{0} & \textcolor{TUMRed}{0} & \textcolor{TUMRed}{0} & \hphantom{\times}&
    \times & \times \tikzmarkend{LRS-first} & \times & \times & \times & \times & \times &\hphantom{\times} &
    \times & \times \tikzmarkend{LRS-second}&\hphantom{\times} &
    \times  & \times & \times & \times & \times & \times & \times & \times \tikzmarkend{LRS-third} \\
    };

    \mymatrixbraceright{G}{1}{1}{Encoding for $\bm_1$}
    \mymatrixbraceright{G}{2}{4}{for $\bm_2$}
    \mymatrixbraceright{G}{5}{6}{for $\bm_3$}
    \mymatrixbraceright{G}{7}{9}{for $\bm_4$}
    \mymatrixbracetop{G}{1}{6}{Encoding at $S_{\cJ_1}$}
    \mymatrixbracetop{G}{8}{14}{at $S_{\cJ_2}$}
    \mymatrixbracetop{G}{16}{17}{at $S_{\cJ_3}$}
    \mymatrixbracetop{G}{19}{26}{at $S_{\cJ_4}$}
    \mymatrixbracebottom{G}{1}{9}{First block of LRS code}
    \mymatrixbracebottom{G}{10}{17}{Second block}
    \mymatrixbracebottom{G}{18}{26}{Third block}
  \end{tikzpicture}
\]
    \caption{Illustration of the support-constrained generator matrix of the $[23,9,15]_{\newinline{4^9}}$ LRS code for the distributed multi-source network.}
    \label{fig:eg-support-constrained-LRS}
  \end{figure}

  Let $q=4, m=9$. We can obtain the support-constrained encoding matrix $\bG$ from a generator matrix of a $[23, 9, 15]_{4^9}$ LRS code with $\ell=3$ blocks. The lengths of blocks are $(n_1, n_2, n_3)=(8,7,8)$.
  Let $\priEle$ be a primitive element of $\F_{4^{9}}$. The block representatives of the LRS code are $(a_1, a_2,a_3)=(1,\priEle,\priEle^2)$ and the column multipliers are $\bb=( 1,\priEle,\dots,\priEle^7, \priEle, \priEle^2,\dots, \priEle^7,\priEle^2,$ $ \priEle^3,\dots, \priEle^9)$.
  Construct a generator matrix $\GLRS$ of the LRS code according to \eqref{eq:Gevaluation} and find a full-rank matrix $\bT\in\F_{4^9}^{9\times 9}$ such that the support-constrained encoding matrix $\bG$ is given by $\bG=\bT\cdot \GLRS$.
  It can be verified by \cref{thm:fieldSizeFromGab} that such a matrix $\bT$ exists over $\F_{4^9}$ and it can be found by solving a linear system of equations.
  For brevity, we omit the explicit solution of $\bT$ here.
\begin{remark}
  With the choice of $\ell$ and $(n_1,\dots,n_\ell)$ for this toy LRS code we intend to show that the number of blocks $\ell$ does not need to be the same as the number of source nodes $s$.
  The value of $\ell$ determines the upper bound in \eqref{eq:sumrank-weight-E} on the sum-rank weight of $\bE$.
  We listed several other parameters of the LRS codes in \cref{tab:ell_grows} that can be used for this network example.
  It can be seen larger the $\ell$ is, larger the error-correction capability required for the LRS code, which results in a larger total length $n$.
  However, larger $\ell$ may result in smaller field size. For instance, the messages $\bm_i$'s are over $\F_{3^{11}}$. According to \cref{tab:ell_grows}, setting $\ell=1$ (i.e., using a distribued Gabidulin code \cite{halbawi2014distributedGab}) requires a field size $q^m=3^{15}$ while using the distributed LRS codes with $\ell=2$ requires a field size $q^m=3^{11}$ (note that the field size of the messages is $3^{11}$).

  \begin{table}[h!]
      \centering
      \begin{tabular}{c|c|c|c|c|c}
      $\ell$ & $q$ & $m$ & $[n,\widetilde{k}, d= 2\ell t+\rho+1]$ & $(n_1,\dots, n_{\ell})$ & $(n_{\cJ_1},n_{\cJ_2},n_{\cJ_3},n_{\cJ_4})$
      \\
      \hline
      $1$ (distributed Gabidulin code) & $2$ & $15$ & $[15, 9, 7]$ & $(15)$ & $(6,1,0,8)$
      \\
      \hline
          $2$ & $3$ & $10$ & $[19, 9 , 11]$ & $(10, 9)$ & $(6,5,0,8)$
      \\
      \hline
      $3$ (\cref{fig:eg-support-constrained-LRS}) & $4$ & $9$ & $[23,9,15]$ & $(8,7,8)$ & $(6,7,2,8)$\\
      \hline
      $4$ & $5$ & $9$ & $[27, 9, 19]$ & $(7,7,7,6)$ & $(6,7,6,8)$\\
      \hline
      $5$ & $7$ & $11$ & $[33, 11, 23]$ & $(7,6,7,7,6)$ & $(8,9,6,10)$\\
      \hline
      $6$ & $7$ & $12$ & $[38, 12, 27]$ & $(7,6,6,7,6,6)$ & $(9,10,8,11)$\\
      \hline
      $7$ & $8$ & $13$ & $[43, 13, 31]$ & $(6,6,6,7,6,6,6)$ & $(10,11,10,12)$\\
    \end{tabular}
      \caption{Resulting parameters of distributed LRS codes for the toy example while increasing $\ell$. The $q$ and $m$ are the minimal value of the required parameters of the field over which the $[n,k,d]$ distributed LRS code can be constructed.}
      \label{tab:ell_grows}
  \end{table}

    In \cref{tab:s_changes}, we list the parameters of LRS codes for several different $\cS=\{\cJ_1,\cJ_2,\cJ_3,\cJ_4\}$. It can be seen that encoding each message independently requires longer code (hence, larger alphabet size) than allowing jointly encoded subsets of all messages.
    \begin{table}[h!]
      \centering
      \begin{tabular}{c|c|c|c|c|c|c}
      $\cS$ &$\ell$ & $q$ & $m$ & $[n,\widetilde{k}, d= 2\ell t+\rho+1]$ & $(n_1,\dots, n_{\ell})$ & $(n_{\cJ_1},n_{\cJ_2},n_{\cJ_3},n_{\cJ_4})$
      \\
      \hline
      \multirow{3}{*}{$\{\{1\},\{2\},\{3\},\{4\}\}$} &$1$& $2$ & $33$ & $[33, 27, 7]$ & $(33)$ & $(7,9,8,9)$\\
      \cline{2-7}
       &$2$& $3$ & $39$ & $[49, 39, 11]$ & $(25,24)$ & $(11,13,12,13)$\\
      \cline{2-7}
       &$3$& $4$ & $51$ & $[65, 51, 15]$ & $(22,22,21)$ & $(15,17,16,17)$\\
      \hline
      \multirow{3}{*}{$\{\{1,2\},\{1,3\},\{2,4\},\{3,4\}\}$} & $1$ & $2$ & $17$ & $[17, 11, 7]$ & $(17)$ & $(6,1,3,7)$\\
      \cline{2-7}
       & $2$ & $3$ & $15$ & $[25, 15, 11]$ & $(13,12)$ & $(10,1,3,11)$\\
      \cline{2-7}
       & $3$ & $4$ & $19$ & $[33, 19, 15]$ & $(11,11,11)$ & $(14,1,3,15)$\\
    \end{tabular}
      \caption{Resulting parameters of distributed LRS codes for the toy example while changing $\cS$. }
      \label{tab:s_changes}
  \end{table}
  \end{remark}
  Now we proceed to apply the \emph{lifting} technique \cite{silva2008rank} to deal with the noncoherent situation.
  Supposing $(\bc_{\cJ_1},\bc_{\cJ_2},\bc_{\cJ_3},\bc_{\cJ_4})= (\bm_1,\bm_2,\bm_3,\bm_4)\cdot \bG$. Each source node $S_{\cJ_i}$ generates $\bC_{\cJ_i}=\extbasis{\bbeta}(\bc_{\cJ_i})\in\Fq^{m\times n_{\cJ_i} }$ by the map defined in \eqref{eq:ext-map} and lifts the $\bC_{\cJ_i}^\top$ by adding the identity and zero matrices as below to obtain the transmitted packets (rows of $\bX$):

  \pgfkeys{tikz/mymatrixenv/.style={decoration={brace},every left delimiter/.style={xshift=5pt, yshift=-1pt},every right delimiter/.style={xshift=-5pt, yshift=-1pt}}}

\pgfkeys{tikz/mymatrix/.style={matrix of math nodes,nodes in empty cells,left delimiter={(},right delimiter={)},inner sep=2pt,outer sep=2pt,column sep=4pt,row sep=4pt,nodes={minimum width=3pt,minimum height=12pt,anchor=center,inner sep=0pt,outer sep=0pt}}}

\pgfkeys{tikz/mymatrixbrace/.style={decorate,thick}}

\tikzset{style green/.style={
    set fill color=TUMGreenDark!80!lime!20,fill opacity=0.3,
    set border color=TUMGreenDark!60!lime!40,draw opacity=1.0,
  },
  style cyan/.style={
    set fill color=cyan!90!blue!60, draw opacity=0.4,
    set border color=blue!70!cyan!30,fill opacity=0.1,
  },
  style orange/.style={
    set fill color=TUMOrange!80,fill opacity=0.3,
    set border color=TUMOrange!90,  draw opacity=0.8,
  },
  style brown/.style={
    set fill color=brown!70!orange!40, draw opacity=0.4,
    set border color=brown, fill opacity=0.3,
  },
  style purple/.style={
    set fill color=violet!90!pink!20, draw opacity=0.5,
    set border color=violet, fill opacity=0.3,
  },
  style black/.style={
    set fill color=none, fill opacity=0.3,
    set border color=black, draw opacity=0.5,
  },
  kwad/.style={
    above left offset={-0.07,0.23},
    below right offset={0.07,-0.23},
    #1
  },
  pion/.style={
    above left offset={-0.07,0.2},
    below right offset={0.07,-0.32},
    #1
  },
  border/.style={
    above left offset={-0.03,0.18},
    below right offset={0.03,-0.3},
    #1
  },
  set fill color/.code={\pgfkeysalso{fill=#1}},
  set border color/.style={draw=#1}
}

\[
  \bX =
  \begin{tikzpicture}[baseline={-0.5ex},mymatrixenv]
    \matrix [mymatrix,inner sep=4pt] (LX)
    {
      \bI_{n_{\cJ_1}} &&&&  \hphantom{holder} \bC_{\cJ_1}^\top \hphantom{holder}  \\
      & \bI_{n_{\cJ_2}} &&& \hphantom{holder}  \bC_{\cJ_2}^\top \hphantom{holder}  \\
      && \bI_{n_{\cJ_3}} && \hphantom{holder}  \bC_{\cJ_3}^\top \hphantom{holder}  \\
      &&& \bI_{n_{\cJ_4}} & \hphantom{holder}  \bC_{\cJ_4}^\top \hphantom{holder}  \\
    };

        \mymatrixbraceright{LX}{1}{4}{$n=\sum_{i=1}^4n_{\cJ_i}$}
        \mymatrixbracebottom{LX}{5}{5}{$m$}
        \mymatrixbracebottom{LX}{1}{4}{$n$}

    draw[thick] (m-4-1.south west) -- (m-4-10.south east);
    \end{tikzpicture}
\]

  Each row is a packet of length $n+m$ ($=23+9=31$ in this toy example) over $\Fq$ ($\F_4$) transmitted into the network.
  Note that for the lifting step, a centralized coordination unit is also needed to instruct the source nodes where to put the identity matrix in their packets.
  \subsection{The General Scheme: Distributed LRS Codes}
  \label{sec:distributed-LRS}
  This section provides a general scheme at the centralized coordination unit to design the overall distributed LRS codes, given:
  \begin{itemize}
  \item the total number of messages $h$ and their lengths $r_1,\dots, r_h$;
  \item the set $\cS=\{\cJ_1, \dots,\cJ_s\}$, where each $\cJ_i\subset [h]$ contains the indices of the messages that the source node $S_{\cJ_i}$ has access to;
  \item the maximum number of malicious nodes $t$ and frozen nodes $\rho$ in the network;
  \item the number of blocks $\ell$ of LRS codes.
  \end{itemize}
  The task is to design the $n_{\cJ}$, for all $\cJ\in\cS$, such that the sink can recover all $h$ messages. The goal of the design is to minimize the total number $n$ of the encoded symbols.

  The general scheme contains the following steps:
  \begin{enumerate}
  \item Solving the following integer linear programming problem for $(n_{\cJ_1},\dots,n_{\cJ_s})$
    \begin{align}
      \text{minimize } \quad & n= n_{\cJ_1}+\dots+n_{\cJ_s}\nonumber \\
      \text{subject to } \quad & \forall \varnothing\neq \cJ'\subseteq[h],\ \sum_{i\in\cJ'} r_i + 2t+\rho\leq n - \sum_{\substack{\cJ\in\cS\\\cJ\subseteq [h]\setminus\cJ'}}n_{\cJ}  \label{eq:capacity-constraints}\ ,\\
                             & \forall \varnothing\neq \Omega\subseteq[h],\ \sum_{\substack{\cJ\in\cS\\ [h]\setminus \cJ\supseteq\Omega} }  n_{\cJ} + \sum_{i\in\Omega} r_i \leq n-2\ell t-\rho \label{eq:dis-zero-constraints} \ ,\\
                             &\forall \cJ\in\cS, \quad\quad\quad n_{\cJ}\geq 0\ .\nonumber
    \end{align}
    \textit{Remark:
    Recall that we assume the min-cut $\mincut_{\cJ'}= n - \sum_{\substack{\cJ\in\cS\\\cJ\subseteq [h]\setminus\cJ'}}n_{\cJ}$, for all $\varnothing\neq\cJ'\subseteq[h]$. With the constraints in \eqref{eq:capacity-constraints}, the choice of $(n_{\cJ_1},\dots, n_{\cJ_s})$ guarantees that the message lengths $(r_1,\dots,r_h)$ are in the capacity region given in \cref{thm:capacity-distributed-network}. \newline
    Let $$\widetilde{k}:=\underset{\varnothing\neq\Omega\subseteq[h]}{\max} \sum\limits_{\substack{\cJ\in\cS\\ [h]\setminus \cJ\supseteq\Omega} }  n_{\cJ} + \sum\limits_{i\in\Omega} r_i\ .$$
    By \cref{thm:subcode}, there exist a subcode of an $[n,\widetilde{k}]$ LRS code whose generator matrix fulfills the support constraints of the encoding matrix $\bG$ for the distributed multi-source network.
    The constraints in \eqref{eq:dis-zero-constraints} guarantees that $\widetilde{k}\leq n-2\ell t-\rho$, which guarantees that the $[n,\widetilde{k}]$ LRS code can decode the rank-metric errors and erasures (see \cref{sec:sum-rank-weight-rank-error}).
    }
  \item Determine the field size $q^m$ required for the $[n, \widetilde{k}]$ LRS code with $\ell$ blocks according to \cref{thm:subcode}. (Tip: The total length should be distributed as evenly as possible into $\ell$ blocks so that the extension degree $m$ is minimized.)
  \item Construct a generator matrix $\GLRS$ of the $[n, \widetilde{k}]_{q^m}$ LRS code according to \eqref{eq:Gevaluation} and \eqref{eq:MooreMatrix}.
  \item Find a full-rank $\bT\in\Fqm^{k\times \widetilde{k}}$ (where $k=\sum_{i=1}^h r_i$) such that the support-constrained encoding matrix $\bG$ can be obtained from $\bG=\bT\cdot \GLRS$. (This can be done by solving a linear system of equations for the entries of $\bT$.)
  \end{enumerate}

\section{Proof of Claim 1}
\label{sec:main_result}

\subsection{Problem Setup}\label{sec:proof_setup}
Let $R_\numMulPolyVar$
be the multivariate commutative polynomial ring as defined in \eqref{eq:multiVarRing}. Note that $R_0=\Fqm$.
Let $\Frobaut$ be the Frobenius automorphism of $R_0$, which we extend to any $a=\sum_{\bi\in\bbN^n_0} a_{\bi}\cdot \beta_{1,1}^{i_1}\cdots \beta_{\ell, n_\ell}^{i_\numMulPolyVar}\in R_\numMulPolyVar$ by
\begin{align*}
    \Frobaut\ :\ R_\numMulPolyVar &\to R_\numMulPolyVar\\
    \sum_{\bi\in\bbN^n_0} a_{\bi}\cdot \beta_{1,1}^{i_1}\cdots \beta_{\ell, n_\ell}^{i_\numMulPolyVar} &\mapsto \sum_{\bi\in\bbN^n_0} \Frobaut(a_{\bi})\cdot \Frobaut(\beta_{1,1}^{i_1})\cdots \Frobaut(\beta_{\ell, n_\ell}^{i_\numMulPolyVar})\ .
\end{align*}
Let $R_\numMulPolyVar[\SkewVar;\Frobaut]$ be the univariate skew polynomial ring with indeterminate $X$, whose coefficients are from $R_\numMulPolyVar$, i.e.,
\begin{align*}
  R_\numMulPolyVar[\SkewVar;\Frobaut] \coloneqq \left\{ \left. \sum_{i=0}^d c_iX^i ~\right|~ d\geq 0, c_0,\dots,c_d \in R_\numMulPolyVar, \right\}\ .
\end{align*}
For ease of notation, when it is clear from the context, we may omit the variable notation in $f(\SkewVar)$ for $f\in\FrobPolysn$, and write only $f$.
The degree of $f=\sum_{i=0}^{d} c_i\SkewVar^{i}\in\FrobPolysn$ is $\deg f=d$ if $d$ is the largest integer such that $c_d\neq 0$. We define $\deg 0 = -\infty$.

Similar to skew polynomials over a finite field, addition is commutative and multiplication is defined using the commutation rule
\begin{align}\label{eq:commutationRule}
  \SkewVar\cdot a = \Frobaut(a)\cdot \SkewVar, \ \forall a\in R_\numMulPolyVar,
\end{align}
which is naturally extended by distributivity and associativity. Just like \eqref{eq:skewProd}, the product of $f,g\in\FrobPolysn$ with $\deg f=d_f$ and $\deg g=d_g$ is
\begin{align}
f\cdot g=\sum_{i=0}^{d_f}\sum_{j=0}^{d_g} f_i\Frobaut^{i}(g_j) \SkewVar^{i+j},\label{eq:prodRn}
\end{align}
and
the degree of the product is $\deg \left(f\cdot g\right) = d_f + d_g$. Note that
in general, $f\cdot g\neq g\cdot f$, for $f,g\in\FrobPolysn$.

By abuse of notation, in the following, we also denote by
\begin{align*}
\locSet
  =&\{a_1\beta_{1,1}^{q-1}, \dots, a_1\beta_{1, n_1}^{q-1}\ \vdots\ \dots \ \vdots \ a_\ell\beta_{\ell,1}^{q-1}, \dots, a_\ell\beta_{\ell, n_\ell}^{q-1}\} \subseteq \mulVarRng
  \end{align*}
  the P-independent set as a subset of $\mulVarRng$.
  Let $\rootSet_i\subseteq \locSet$ be the set as in \eqref{eq:rootSet} corresponding to $\zeroSet_i$ and $f_{\rootSet_i}\in\FrobPolysn$ be the minimal polynomial of $\rootSet_i$. 
We note the following properties of $\FrobPolysn$, which will be useful for the proof of the main result in \cref{sec:ind_proof}. The detailed proofs of these properties can be found in \cref{appendix:SkewPolyProperty}.
\begin{enumerate}[label={\bfseries P\arabic*}]
\item \label{p:noZeroDiv}
  $\FrobPolysn$ is a ring without zero divisors.
\item \label{p:gcrd}
  For any sets 
  $\rootSet_1, \rootSet_2\subseteq \mulVarRng$ s.t. $\rootSet_1\cup\rootSet_2$ is P-independent,
  let $f_1, f_2\in \FrobPolysn$ be the minimal polynomials of $\rootSet_1, \rootSet_2$, respectively. Then the greatest common right divisor $\gcrd(f_1,f_2)$ is the minimal polynomial of $\rootSet = \rootSet_1\cap \rootSet_2$, denoted by $f_{\rootSet_1\cap\rootSet_2}$.
  In particular, $\rootSet_1\cap\rootSet_2=\varnothing \iff \gcrd(f_1,f_2)=1$.
\item \label{p:xlrdivision}
  For $t\in\bbN$ and any $f\in\FrobPolysn$, $\SkewVar^t |_l f \iff \SkewVar^t|_r f$.
  In this case, we may write just $\SkewVar^t| f$.
\item \label{p:xdivision}
  For $t\in\bbN$ and any $f_1, f_2 \in\FrobPolysn$ such that $\SkewVar\not\divides f_2$, then $\SkewVar^t| (f_1\cdot f_2)$ if and only if $\SkewVar^t| f_1$.
\end{enumerate}
In the main result in \cref{thm:sufficientCond}, we are interested in skew polynomials in the following form: for any $\zeroSet\subseteq[n], \tfzt\geq 0$
\begin{align}\label{eq:fZt}
  \fZt(\zeroSet,\tfzt)\coloneqq \SkewVar^{\tfzt}\cdot \lclm_{\substack{\alpha\in\{a_l\betalt^{q-1}  |\\ \indMap(l,t)\in\zeroSet\}}}\{(\SkewVar-\alpha)\} \in \FrobPolysn\ ,
\end{align}
where $\indMap(l,t)$ is as defined in \eqref{eq:indMap}.

Define the set of skew polynomials of this form:
\begin{equation}
  \label{eq:Skewnk}
  \begin{split}
  \Skewnk \coloneqq \{&\fZt(\zeroSet,\tfzt)\ |\ \tfzt\geq 0, \zeroSet\subseteq [n]\\&\text{s.t. }|\zeroSet|+\tfzt\leq k-1\}\subseteq \FrobPolysn.
  \end{split}
\end{equation}
Note that $\deg f \leq k-1, \forall f\in\Skewnk$.
We also note the following properties of polynomials in $\Skewnk$, whose proofs are given in \cref{appendix:SkewPolyProperty}.
\begin{enumerate}[resume,label={\bfseries P\arabic*}]
\item \label{p:gcrd2polys}
  For any $f_1 = \fZt(\zeroSet_1,\tfzt_1), f_2 = \fZt(\zeroSet_2, \tfzt_2)\in\Skewnk$, we have
  \begin{align*}
    \gcrd(f_1,f_2)=\fZt(\zeroSet_1\cap \zeroSet_2, \min\{\tfzt_1,\tfzt_2\})\in\Skewnk\ .
  \end{align*}
\item \label{p:reducevar}
  Let $f=\fZt(\zeroSet,\tfzt)\in\Skewnk$ and let $f'=f|_{\beta_{\ell, n_\ell}=0}\in R_{\numMulPolyVar-1}[\SkewVar;\Frobaut]$ (namely, we substitute $\beta_{\ell, n_\ell}=0$ in each coefficient of $f$). Then $f'\in\cS_{n-1,k}$ and
\begin{align*}
    f'= \begin{cases}
    \fZt(\zeroSet,\tfzt) & n\not\in \zeroSet,\\
    \fZt(\zeroSet\setminus\{n\}, \tfzt+1) & n\in \zeroSet.
    \end{cases}
\end{align*}
\end{enumerate}
\subsection{Main Result}\label{sec:ind_proof}
The following theorem is a more general statement than \cref{claim} and it is the analog of \cite[Theorem 3.A]{yildiz2019gabidulin} for skew polynomials.
\begin{theorem}\label{thm:sufficientCond}
  Let $k\geq \numRows\geq 1$ and $n\geq 0$. For any $f_1, f_2,\dots, f_\numRows\in \Skewnk$, the following are equivalent:
  \begin{enumerate}[label={\rm(\roman*)}]
  \item For all $g_1,g_2,\dots, g_\numRows\in\FrobPolysn$ such that $\deg (g_i\cdot f_i)\leq k-1$, we have \label{item:equiv1}
    \begin{align*}
      \sum_{i=1}^\numRows g_i\cdot f_i = 0\quad \implies \quad g_1=g_2=\dots = g_\numRows =0\ .
    \end{align*}
  \item For all nonempty $\Omega\subseteq [\numRows]$, we have \label{item:equiv2}
    \begin{align}\label{eq:equiv2}
      k-\deg (\gcrd_{i\in\Omega} f_i) \geq \sum_{i\in\Omega} (k-\deg f_i)\ .
    \end{align}
  \end{enumerate}
\end{theorem}

The proof of \cref{thm:sufficientCond} is given in \versionShortLong{\cite[Appendix C]{full_version}}{\cref{sec:proofMainResult}}. We will show in \cref{cor:claim} that \cref{claim} is a special case of \cref{thm:sufficientCond}.
For this purpose, we give an equivalent way of writing it in terms of matrices with entries from $R_\numMulPolyVar$. This is done in \cref{thm:sufficientMat}, which is an analog to \cite[Theorem 3.B]{yildiz2019gabidulin}.

We first describe the multiplication between skew polynomials in matrix language.
Let $\up=\sum_{i} \upc_i \SkewVar^{i}\in\FrobPolysn$. For $b-a\geq \deg \up$, define the following matrix in $R_n^{a\times b}$
\begin{align*}
  \bS_{a\times b}(\up)\coloneq \nonumber 
  \begin{pmatrix}
    \upc_0 & \cdots & \upc_{b-a} &\vphantom{v} \\
    & \Frobaut(\upc_0) & \cdots & \Frobaut(\upc_{b-a})\\
    & \ddots &  \ddots &\ddots\\
    &&\Frobaut^{a-1}(\upc_0) &  \cdots  &\Frobaut^{a-1}(\upc_{b-a})
 \end{pmatrix}.
\end{align*}
  In particular, for $a=1$, denote by $\FrobPolysn_{<b}$ the set of skew polynomials of degree strictly less than $b$. The map
  \begin{equation}
    \label{eq:polyMatMap}
    \begin{split}
      \bS_{1\times b}(\cdot)\ : \ \FrobPolysn_{<b} &\to \mulVarRng^b\\
      u &\mapsto (u_0, \dots, u_{b-1})
    \end{split}
  \end{equation}
  is bijective and $\bS_{1\times b}(0)=\0, \forall b\in\bbN$.
For any skew polynomial $\vp=\sum_i\vpc_i\SkewVar^{i}\in\FrobPolysn$, we have
\begin{align}\label{eq:matOfProduct}
  \bS_{a\times b}(\vp\cdot \up) = \bS_{a\times c}(\vp)\cdot \bS_{c\times b}(\up)\ ,
\end{align}
where $a,b,c\in\bbN$ are such that $c-a\geq \deg  \vp, b-c\geq \deg  \up$.
As a special case, when $\vp = \SkewVar^\tfzt, \tfzt\in\bbN$, we can write
\begin{align}
  \bS_{a\times(b+\tfzt)}(\SkewVar^\tfzt\cdot \up)=& \bS_{a\times (a+\tfzt)}(\SkewVar^\tfzt)\cdot \bS_{(a+\tfzt)\times (b+\tfzt)}(\up)\nonumber\\
  =&
      \left( \0_{a\times \tfzt}\quad \bI_{a \times a}\right)
                           \cdot \bS_{(a+\tfzt)\times (b+\tfzt)}(\up)  \label{eq:matfZt}. 
\end{align}
By the definition in \eqref{eq:fZt}, $\fZt(\zeroSet,\tfzt)=\SkewVar^\tfzt\cdot \up$ for some $\up\in\FrobPolysn$. It can be readily seen from \eqref{eq:matfZt} that the first $\tfzt$ columns of $\bS_{a\times (b+\tfzt)}(\fZt(\zeroSet,\tfzt))$ are all zero.

For $s\in[k]$, $i\in [\numRows]$, let $f_i=\fZt(\zeroSet_i, \tfzt_i)\in\Skewnk$.
We write $\bS(f_i)$ instead of $\bS_{(k-\tfzt_i-|\zeroSet_i|)\times k}(f_i)$ for ease of notation. By \eqref{eq:matfZt}, $\bS(f_i)$ looks like
\begingroup
\setlength\arraycolsep{3pt}
\begin{align*}
\bS(f_i)
	 &=\begin{pmatrix}
        0 & \cdots & 0 & \times & \times & \cdots & \times\\
        0 & \cdots & 0 &  & \times & \times & \cdots & \times\\
        \vdots && \vdots &&&\ddots & \ddots & &\ddots\\
        \makebox[0pt][l]{$\smash{\underbrace{\phantom{\begin{matrix}0&\cdots&0\end{matrix}}}_{\tfzt_i}}$}0 & \cdots & 0&
        \makebox[0pt][l]{$\smash{\hspace{-5pt}\underbrace{\phantom{\begin{matrix}\times&\times&\cdots\end{matrix}}}_{k-1-\tfzt_i-|\zeroSet_i|}}$}
        &&&
        \makebox[0pt][l]{$\smash{\underbrace{\phantom{\begin{matrix}\times&\times&\cdots&\times\end{matrix}}}_{|\zeroSet_i|+1}}$}
        \times & \times & \cdots & \times\\
    \end{pmatrix}\!\!\!\!\left.\vphantom{\begin{pmatrix}0\\0\\\vdots\\0\end{pmatrix}}
				\right\} {\scriptstyle k-\tfzt_i-|\zeroSet_i|}\\
\end{align*}
\endgroup
where the $\times$'s represent possibly non-zero entries.
Then, applying \eqref{eq:matOfProduct} to the expression $g_i\cdot f_i$ in \cref{thm:sufficientCond} yields
\begin{align}
  \bS_{1\times k}(g_i\cdot f_i) = \bu_i \cdot \bS(f_i)\ ,\nonumber
\end{align} 
where $\bu_i=\bS_{1\times (k-\tfzt_i-|\zeroSet_i|)}(g_i)$ is a row vector.
Therefore, we can write
\begin{align}
  \bS_{1\times k}(\sum_{i=1}^\numRows g_i\cdot f_i) =
    (\bu_1, \cdots , \bu_\numRows)
                     \cdot
                     \underbrace{\begin{pmatrix}
                       \bS(f_1)\\ \vdots \\ \bS(f_\numRows)
                     \end{pmatrix}}_{=:\bM(f_1,\dots, f_\numRows)}\label{eq:defMmat}
\end{align}
which is a linear combination of the rows of $\bM(f_1,\dots,f\numRows)$.

The following theorem is an equivalent statement to \cref{thm:sufficientCond}, in matrix language.
\begin{theorem}\label{thm:sufficientMat}
  Let $k\geq s\geq 1$ and $n\geq 0$. For $i\in [\numRows]$, let $\zeroSet_i\in[n], \tfzt_i\geq 0$ such that $\tfzt_i+|\zeroSet_i|\leq k-1$ and $f_i=\fZt(\zeroSet_i, \tfzt_i)\in\Skewnk$. The matrix $\bM(f_1,\dots, f_\numRows)$ defined in \eqref{eq:defMmat} has full row rank if and only if, for all nonempty $\Omega\subseteq[\numRows]$,
  \begin{align}
    k-\left|\bigcap_{i\in\Omega}\zeroSet_i\right| - \min_{i\in\Omega} \tfzt_i \geq \sum_{i\in\Omega}(k-\tfzt_i-|\zeroSet_i|)\ .
    \label{eq:matEquiv2}
  \end{align}
\end{theorem}
\begin{IEEEproof}
    For brevity, we write $\bM$ instead of $\bM(f_1,\dots, f_\numRows)$ in the proof.
    The logic of the proof is as following
    \begin{align*}
      \bM \text{ has full row rank }\overset{\matProne}{\iff} \ref{item:equiv1} \overset{\text{\cref{thm:sufficientCond}}}{\iff} \ref{item:equiv2} \overset{\matPrtwo}{\iff} \eqref{eq:matEquiv2}\text{ holds}
    \end{align*}
    where \ref{item:equiv1} and \ref{item:equiv2} are shown to be equivalent in \cref{thm:sufficientCond}.
    We only need to show the equivalence \matProne{} and \matPrtwo.

    \matProne: Assuming $\bM$ has full row rank, it is equivalent to write
  \begin{equation}
    \begin{split}
      \forall &\bu
      \in\mulVarRng^{1\times \sum_{i=1}^s(k-\tfzt_i-|\zeroSet_i|)},
      \bu\cdot \bM=\0 \Longrightarrow \bu=\0\ .
  \end{split}
    \label{eq:matEquiv1}
  \end{equation}
  Partition $\bu$ into $s$ blocks $(\bu_1,\dots, \bu_\numRows)$, where $\bu_i\in\mulVarRng^{1\times (k-\tfzt_i-|\zeroSet_i|)}$. Note that $\bu=\0\iff \forall i\in[s], \bu_i =\0$.
  For each $i\in[s]$, the set $\{g_i\ |\ \bS_{1\times (k-\tfzt_i-|\zeroSet_i|)}(g_i)=\bu_i, \forall \bu_i\in\mulVarRng^{1\times (k-\tfzt_i-|\zeroSet_i|)}\}$ is $\FrobPolysn_{<(k-\tfzt_i-|\zeroSet_i|)}$, which is the set of skew polynomials of degree less than $ (k-\tfzt_i-|\zeroSet_i|)$, since the map $\bS_{1\times *}(\cdot)$ as in \eqref{eq:polyMatMap} is bijective. Therefore, $\bu_i=\0\iff g_i=0, \forall i\in[\numRows]$. It can be further inferred that every $\bu\in \mulVarRng^{1\times \sum_{i=1}^{\numRows}(k-\tfzt_i-|\zeroSet_i|)}$ corresponds to a unique tuple $(g_1, \dots, g_\numRows)\in \FrobPolysn_{<(k-\tfzt_1-|\zeroSet_1|)}\times \dots\times\FrobPolysn_{< (k-\tfzt_\numRows-|\zeroSet_\numRows|)}$. We denote the Cartesian product by $\cG$.
  Since $\deg f_i= \tfzt_i+|\zeroSet_i|,\forall i\in[\numRows]$, for any tuple $(g_1,\dots, g_\numRows)\in\cG$, $\deg (g_i\cdot f_i) \leq k-1, \forall i\in[\numRows]$.

  By the equality in \eqref{eq:defMmat}, $\bu\cdot \bM=\bS_{1\times k}(\sum_{i=1}^\numRows g_i\cdot f_i)$ and $\bS_{1\times k}(\sum_{i=1}^\numRows g_i\cdot f_i)=\0 \iff \sum_{i=1}^\numRows g_i\cdot f_i=0$.
  Hence it is equivalent to write \eqref{eq:matEquiv1} as:
  \begin{align*}
    \forall g_1,\dots,g_\numRows\in\FrobPolysn &\text{ such that }\deg (g_i\cdot f_i)\leq k-1\ ,\\
    &  \sum_{i=1}^\numRows g_i\cdot f_i=0 \implies g_i=0, \forall i\in[\numRows]\ ,
  \end{align*}
  which is exactly the statement \ref{item:equiv1}.
\matPrtwo:
  It follows from \ref{p:gcrd2polys} that for any nonempty set $\Omega\subseteq [\numRows]$,
  \begin{align*}
    \deg (\gcrd_{i\in\Omega} f_i) = \left|\bigcap_{i\in\Omega}\zeroSet_i\right| + \min_{i\in\Omega} \tfzt_i\ .
  \end{align*}
  Then the left hand side of \eqref{eq:equiv2} $k-\deg (\gcrd_{i\in\Omega} f_i)= k- |\bigcap_{i\in\Omega}\zeroSet_i| - \min_{i\in\Omega} \tfzt_i$, which is the left hand side of \eqref{eq:matEquiv2}. By the definition of $f_i=\fZt(\zeroSet_i,\tfzt_i)$ in \eqref{eq:fZt}, the right hand side of \eqref{eq:equiv2} is
    $\sum_{i\in\Omega}(k-\deg f_i) = \sum_{i\in\Omega}(k-(|\zeroSet_i|+\tfzt_i))$, which is the right hand side of \eqref{eq:matEquiv2}.
\end{IEEEproof}

As a special case, when $s=k, \tfzt_i=0$ and $|\zeroSet_i| = k-1, \forall i\in[k]$, each block $\bS(f_i)$ becomes a row vector with entries being the coefficients of $f_i=\fZt(\zeroSet_i,\tfzt_i)=\sum_{j=0}^{k-1} f_{i,j+1}\SkewVar^{j}\in\FrobPolysn$ and
\begin{align}
  \bM(f_1,\dots,f_k) =
  \begin{pmatrix}
    f_{11} & f_{12} & \cdots & f_{1k}\\
    f_{21} & f_{22} & \cdots & f_{2k}\\
    \vdots & \vdots & \ddots & \vdots\\
    f_{k1} & f_{k2} & \cdots & f_{kk}
  \end{pmatrix}\in\mulVarRng^{k\times k}\ .\label{eq:defMk}
\end{align}
Note that $\bM(f_1, \dots,f_k)$ coincides with the matrix $\bT$ in \eqref{eq:defTmat}.
Hence we have \cref{cor:claim} below, which is \cref{claim}.

\begin{corollary}\label{cor:claim}
For $i\in[k]$, let $\zeroSet_i\subseteq [n]$ with $|\zeroSet_i|=k-1$. Then $\det \bM(f_1,\dots,f_k)$ is a nonzero polynomial in $\mulVarRng$ if and only if for all nonempty $\Omega\subseteq[k]$,
  $k-\left |\bigcap_{i\in\Omega} \zeroSet_i \right| \geq |\Omega|. $
\end{corollary}

\begin{new}
  \section{Conclusion and Outlook}
  In this work, we extended the previous work on MDS codes and MRD codes with a support-constrained generator matrix
  to MSRD codes.
  We first investigated the minimum required field size to construct an MSRD code (particularly an LRS code) with a support-constrained generator matrix. For this purpose, we proved that the conditions on the support constraints such that an MDS/MRD code exists are also the necessary and sufficient conditions for the existence of an MSRD code via the framework of skew polynomials.
For any support constraints fulfilling the conditions, an $[n,k]_{q^m}$ LRS code with a support-constrained generator matrix exists, for any prime power $q\geq \ell+1$ and integer $m\geq \max_{l\in[\ell]}\{k-1+\log_qk, n_l\}$, where $\ell$ is the number of blocks and $n_l$ is the length of the $l$-th block of the LRS code.
  With these results, we proposed an application of LRS codes with a support-constrained generator matrix in the distributed multi-source networks, where a collection of messages is to be sent via a linearly coded network with unknown topology, and each source node only has access to a subset of the messages. The provided LRS codes with a support-constrained generator matrix enable the correction of errors and erasures that occur in the network.

  A recent work \cite{martinez2022generalMSRD} has provided new constructions of MSRD codes that generalize LRS codes. Several constructions require smaller field sizes than LRS codes. Investigating these new constructions for the support constraints may result in a smaller required field size to obtain an MSRD code with a support-constrained generator matrix.
\end{new}
\appendices
\section{Proofs of Properties of Skew Polynomials}

\label[appendix]{appendix:SkewPolyProperty}

\ref{p:noZeroDiv}:
  $\FrobPolysn$ is a ring without zero divisors.
\begin{IEEEproof}
The ring properties of $\FrobPolysn$ are trivial, we only need to show that it has no zero divisors.

Note that for any $a,b\in R_\numMulPolyVar$, $\Frobaut(a+b)=\Frobaut(a) + \Frobaut(b)$.
It can be seen from \eqref{eq:prodRn} that if $f,g\neq 0$, then $f\cdot g\neq 0$ since the leading coefficients of $f_{d_f}, g_{d_g}$ are nonzero and therefore $f_{d_f}\Frobaut^{d_f}(g_{d_g})$ is nonzero. Hence, $\FrobPolysn$ does not have zero divisors.
\end{IEEEproof}

\ref{p:gcrd}:
For any sets
$\rootSet_1, \rootSet_2\subseteq \mulVarRng$ s.t. $\rootSet_1\cup\rootSet_2$ is P-independent,
let $f_1, f_2\in \FrobPolysn$ be the minimal polynomials of $\rootSet_1, \rootSet_2$, respectively. Then the greatest common right divisor $\gcrd(f_1,f_2)$ is the minimal polynomial of $\rootSet = \rootSet_1\cap \rootSet_2$, denoted by $f_{\rootSet_1\cap\rootSet_2}$.
In particular, $\rootSet_1\cap\rootSet_2=\varnothing \iff \gcrd(f_1,f_2)=1$.
\begin{IEEEproof}
  We can write the minimal polynomial of the set $\rootSet_1\cap\rootSet_2$ by the least common left multiplier as in \eqref{eq:lclm}, i.e., $f_{\rootSet_1\cap\rootSet_2} = \underset{a\in\rootSet_1\cap\rootSet_2}{\lclm}\{(\SkewVar-a)\}$. Then we can write
  $f_1 = g_1\cdot \underset{a\in\rootSet_1\cap\rootSet_2}{\lclm}\{(\SkewVar-a)\}$
  and
  $f_2 = g_2 \cdot\underset{a\in\rootSet_1\cap\rootSet_2}{\lclm}\{(\SkewVar-a)\},$ for some $g_1,g_2\in\FrobPolysn$.
  Therefore, it is clear that $f_{\rootSet_1\cap\rootSet_2}\mid \gcrd(f_1,f_2)$.
  
  Now we only need to show that $\deg f_{\rootSet_1\cap\rootSet_2}=\deg \gcrd(f_1,f_2)$.
  Since $\rootSet_1\cup\rootSet_2$ is P-independent, $\rootSet_1,\rootSet_2$ and $\rootSet_1\cap\rootSet_2$ are also P-independent. Then $\deg f_{\rootSet_1\cup\rootSet_2} = |\rootSet_1\cup\rootSet_2|$, $\deg f_{\rootSet_1}=|\rootSet_1|$, $\deg_{\rootSet_2}=|\rootSet_2|$ and $\deg f_{\rootSet_1\cap\rootSet_2} = |\rootSet_1\cap\rootSet_2|$.
  It follows from \cite[Proposition 5.12]{gluesing2021introduction} that the minimal polynomial of $\rootSet_1\cup\rootSet_2$ is
  \begin{align*}
    f_{\rootSet_1\cup\rootSet_2} = \lclm(f_1, f_2)
  \end{align*}
  and from \cite[Eq.(24)]{ore1933theory} that
  \begin{align*}
    \deg\gcrd(f_1,f_2)=&\deg f_1+\deg f_2-\deg \lclm(f_1,f_2) \\
    =& |\rootSet_1|+|\rootSet_2|-|\rootSet_1\cup\rootSet_2|\\
    =& |\rootSet_1\cap\rootSet_2|= \deg f_{\rootSet_1\cap\rootSet_2}
  \end{align*}
\end{IEEEproof}

\ref{p:xlrdivision}:
For $t\in\bbN$ and any $f\in\FrobPolysn$, $\SkewVar^t |_l f \iff \SkewVar^t|_r f$.
In this case, we may write just $\SkewVar^t| f$.
\begin{IEEEproof}
  If $\SkewVar^t|_l f$, then with some $g\in\FrobPolysn$ we can write $f= \SkewVar^t\cdot g = \FrobautPolyt{g}{t}\cdot \SkewVar^t$, where $\FrobautPolyt{g}{t}=\sum_i \Frobaut^t(g_i)\SkewVar^i$. Then it is obvious that $X^t|_rf$. Similarly, if $\SkewVar^t |_r f$, we can write $f = g\cdot \SkewVar^t = \SkewVar^t\cdot \FrobautPolyt{g}{-t}$ and it is obvious that $X^t|_l f$. This property has been also shown in \cite[Theorem 7]{ore1933theory}.
\end{IEEEproof}

\ref{p:xdivision}:
  For $t\in\bbN$ and any $f_1, f_2 \in\FrobPolysn$ such that $\SkewVar\not\divides f_2$, then $\SkewVar^t| (f_1\cdot f_2)$ if and only if $\SkewVar^t| f_1$.
\begin{IEEEproof}
  We first show $\SkewVar^t| (f_1\cdot f_2)\Longleftarrow \SkewVar^t| f_1$. Suppose $\SkewVar^t| f_1$, then we can write $f_1 = \SkewVar^t\cdot f_1'$ with some $f_1'\in\FrobPolysn$. Then $f_1\cdot f_2= \SkewVar^t\cdot f_1'\cdot f_2$ and it can be seen that $\SkewVar^t|_l (f_1\cdot f_2)$. By \ref{p:xlrdivision}, we have $\SkewVar^t|(f_1\cdot f_2)$.

  For the other direction, we first show that $\SkewVar| (f_1\cdot f_2)\implies\SkewVar| f_1$ by contradiction.
  Assume $\SkewVar\not\divides f_1$, then we can write $f_1= f_1'+a$ with some $f_1'\in\FrobPolysn$ such that $\SkewVar|f_1'$ and $a\in R_\numMulPolyVar \setminus\{0\}$.
  Since $\SkewVar\not\divides f_2$, we can write $f_2 =f_2'+b$, with some $f_2'\in\FrobPolysn$ such that $\SkewVar|f_2'$ and $b\in R_\numMulPolyVar \setminus\{0\}$. Then,
  \begin{align*}
    f_1\cdot f_2 &= (f_1'+a)(f_2'+b)\\
                 &= f_1'\cdot f_2' +a\cdot f_2' + f_1'\cdot b+a\cdot b
  \end{align*}
  where the first three summands are all divisible by $\SkewVar$ but $a\cdot b\neq 0$ (since $R_\numMulPolyVar $ is a ring without zero divisor) and $\SkewVar\not\divides a\cdot b$. This implies $\SkewVar\not\divides (f_1\cdot f_2)$, which is a contradiction.
  We can then extend the following steps $t$ times and the statement is proven.
  Note that $\SkewVar^2|(f_1\cdot f_2) \implies \SkewVar | (f_1\cdot f_2) \implies \SkewVar|f_1$. Write $f_1=\SkewVar\cdot g$ with some $g\in\FrobPolysn$, then
  \begin{align*}
    \SkewVar^2|(f_1\cdot f_2) &\implies \SkewVar | (g\cdot f_2) \underset{\SkewVar\not\divides f_2}{\implies} \SkewVar|g \\
    &\implies (\SkewVar\cdot \SkewVar)|(\SkewVar\cdot g)\implies \SkewVar^2| f_1\ .
  \end{align*}
\end{IEEEproof}
For any $\zeroSet_i\subseteq[n],i\in[k], l\in[\ell]$, we
denote $\zeroSet_i^{(l)}\coloneq\{t\ |\ \indMap(l,t) \in \zeroSet_i\}$ and $\rootSet_i^{(l)}=\{ a_l \betalt^{q-1} \ |\ t \in \zeroSet_i^{(l)}\}$.
We need some properties of the set of roots of skew polynomials in order to prove \ref{p:gcrd2polys}. 
It follows from \cref{lem:noMoreZero} that $f_{\rootSet_i}$ only vanishes on $\rootSet_i$ while evaluating on $\locSet$. The following lemma gives the structure of the roots of $f_{\rootSet_i}$ while evaluating on $\mulVarRng$.
\begin{lemma}[{\cite[Theorem 4]{liu2015construction}}]
  \label{lem:structureRootsBlock}
  For $l=1,\dots, \ell$, let $f_i^{(l)}$ be the minimal polynomial of $\rootSet_i^{(l)}$ and
  $\overline{\rootSet_i^{(l)}}\coloneq\{\alpha\in\mulVarRng ~|~ f_i^{(l)}(\alpha)=0\}$.
  Then, for all $l=1,\dots, \ell$,
  \begin{align}
      \overline{\rootSet_i^{(l)}} &= \{a_l\beta^{q-1}\ |\ \beta\in \myspan{\betalt}_{t\in\zeroSet_i^{(l)}}\setminus \{0\}\}\subseteq C_\Frobaut(a_l)\label{eq:rootSetBarl}\\
      |\overline{\rootSet_i^{(l)}}|&=q^{|\zeroSet_i^{(l)}|}-1 \label{eq:rootSetBarlSize}
  \end{align}
  where $C_\Frobaut(a_l)$ is the $\Frobaut$-conjugacy class of $a_l$ as defined in \cref{def:conjugacyClasses}.
  \end{lemma}
  \begin{theorem}
  \label{thm:rootsStructure}
  Let $\rowpoly_{i}$ be the minimal polynomial of $\rootSet_i$. Denote the set of roots of $f_i$ while evaluating on $\mulVarRng$ by $\overline{\rootSet_i}\coloneq\{\alpha\in\mulVarRng ~|~ f_i(\alpha)=0\}$. Then
\begin{align}
\label{eq:rootsZi}
  \overline{\rootSet_i} &=  \bigcup_{l=1}^{\ell} \overline{\rootSet_i^{(l)}}, \text{ where }\overline{\rootSet_i^{(l)}} \text{ is as in }\eqref{eq:rootSetBarl}\\
  |\overline{\rootSet_i}| &= \sum_{l=1}^{\ell}|\overline{\rootSet_i^{(l)}}| = \sum_{l=1}^{\ell}q^{|\zeroSet_i^{(l)}|}-\ell.
\end{align}
\end{theorem}
\begin{IEEEproof}
Note that for all $l=1,\dots, \ell$, $\rootSet_i^{(l)}$ are P-independent and they are from different conjugacy classes. It follows from \cite[Corollary 4.4]{LamLer2004} that for such sets,
$\overline{\bigcup_{l=1}^{\ell} \rootSet_i^{(l)}}=\bigcup_{l=1}^{\ell} \overline{\rootSet_i^{(l)}}$.
\end{IEEEproof}
It is clear that the $\alpha$'s in \eqref{eq:fZt} are P-independent. It follows from \cref{def:PindSet} that $\deg \fZt(\zeroSet,\tfzt) = |\zeroSet|+\tfzt$. By \cref{thm:rootsStructure}, the set of roots of $\fZt(\zeroSet,\tfzt)$ is
\begin{align}
\label{eq:fZtRoots}
\{0\}^{\tfzt}\cup\bigcup_{l=1}^{\ell}\{a_l\beta^{q-1}\ |\ \beta\in \myspan{\betalt}_{t\in\zeroSet^{(l)}}\setminus\{0\}\}
\end{align}
where $\zeroSet^{(l)}=\{t\ |\ \indMap(l,t)\in\zeroSet\}$. The notation $\{0\}^{\tfzt}$ is to imply that $\SkewVar^\tfzt\divides \fZt(\zeroSet, \tfzt)$ and $X^{\tau+1}\nmid \fZt(\zeroSet,\tfzt)$.

\ref{p:gcrd2polys}:
  For any $f_1 = \fZt(\zeroSet_1,\tfzt_1), f_2 = \fZt(\zeroSet_2, \tfzt_2)\in\Skewnk$, we have
  \begin{align*}
    \gcrd(f_1,f_2)=\fZt(\zeroSet_1\cap \zeroSet_2, \min\{\tfzt_1,\tfzt_2\})\in\Skewnk\ .
  \end{align*}
\begin{IEEEproof}
  We prove the property by showing that the skew polynomials on both side have the same set of roots.
  Denote by $\widebar{\rootSet_1}, \widebar{\rootSet_2}, \widebar{\rootSet_{1,2}}\subseteq \mulVarRng$ the set of all roots in $\mulVarRng$ of $f_1,f_2, \fZt(\zeroSet_1\cap \zeroSet_2, \min\{\tfzt_1,\tfzt_2\})$, respectively.
  By the structure of roots of $\fZt(\zeroSet,t)$ given in \eqref{eq:fZtRoots},
  \begin{align*}
    \widebar{\rootSet_i} &=\{0\}^{\tfzt_i}\cup\bigcup_{l=1}^{\ell}\{a_l\beta^{q-1}\ |\ \beta\in \myspan{\betalt}_{t\in\zeroSet_i^{(l)}}\setminus\{0\}\},\quad i=1,2\\
    \widebar{\rootSet_{1,2}} &= \{0\}^{\min\{\tfzt_1,\tfzt_2\}}\cup\bigcup_{l=1}^{\ell}\{a_l\beta^{q-1}\ |\ \beta\in \myspan{\betalt}_{t\in\zeroSet_{1,2}^{(l)}}\setminus\{0\}\}
  \end{align*}
  where $\zeroSet_i^{(l)}\coloneq\{t\ |\ \indMap(l,t)\in\zeroSet_i\}$ and $\zeroSet_{1,2}^{(l)}\coloneq\{t\ |\ \indMap(l,t)\in\zeroSet_{1}\cap\zeroSet_2\}$.
  The set of roots of $\gcrd(f_1,f_2)$ is
  \begin{align*}
    \widebar{\rootSet_1}\cap \widebar{\rootSet_2} &= \{0\}^{\min\{\tfzt_1,\tfzt_2\}}\\
    &\qquad \cup\bigcup_{l=1}^{\ell}\{a_l\beta^{q-1}\ |\ \beta\in \myspan{\betalt}_{t\in\zeroSet_1^{(l)}\cap\zeroSet_2^{(l)}}\setminus\{0\}\}\ .
  \end{align*}
  It can be seen that $\zeroSet_1^{(l)}\cap\zeroSet_2^{(l)} = \zeroSet_{1,2}^{(l)}, \forall l\in[\ell]$. Hence, $\widebar{\rootSet_1}\cap \widebar{\rootSet_2}=\widebar{\rootSet_{1,2}}$.
\end{IEEEproof}

\ref{p:reducevar}:
  Let $f=\fZt(\zeroSet,\tfzt)\in\Skewnk$ and let $f'=f|_{\beta_{\ell, n_\ell}=0}\in R_{\numMulPolyVar-1}[\SkewVar;\Frobaut]$ (namely, we substitute $\beta_{\ell, n_\ell}=0$ in each coefficient of $f$). Then $f'\in\cS_{n-1,k}$ and
\begin{align*}
    f'= \begin{cases}
    \fZt(\zeroSet,\tfzt) & n\not\in \zeroSet,\\
    \fZt(\zeroSet\setminus\{n\}, \tfzt+1) & n\in \zeroSet.
    \end{cases}
\end{align*}
\begin{IEEEproof}
  Denote by $\rootSet$ the subset of $\locSet$ corresponding to $\zeroSet$ as in \eqref{eq:rootSet}.
  It is trivial that $f'\in \cS_{n-1,k}$ and $f'=\fZt(\zeroSet,\tfzt)$ when $n\not\in \zeroSet$. Suppose $n\in \zeroSet$, then $a_\ell\beta_{\ell,n_\ell}^{q-1}\in\rootSet$.
  Let $g = \underset{\alpha\in\rootSet\setminus\{ a_\ell\beta_{\ell,n_\ell}\}}{\lclm}\{\SkewVar-\alpha\}$, then
  \begin{align*}
      f'=&\SkewVar^\tfzt \cdot (\underset{\alpha\in\rootSet}{\lclm}\{\SkewVar-\alpha\}) |_{\beta_{\ell,n_\ell}=0}\\
      =&\SkewVar^\tfzt \cdot\left(\left(\SkewVar-(a_\ell\beta_{\ell,n_\ell}^{q-1})^{g(a_\ell\beta_{\ell,n_\ell}^{q-1})}\right)\cdot g \right)|_{\beta_{\ell,n_\ell}=0}\\
      =&\SkewVar^\tfzt \cdot \SkewVar\cdot g \\
      =&\SkewVar^{\tfzt+1} \cdot g \\
      =&\SkewVar^{\tfzt+1} \cdot(\underset{\alpha\in\rootSet\setminus\{a_\ell\beta_{\ell,n_\ell}^{q-1} \}}{\lclm}\{\SkewVar-\alpha\}) \\
      =&\fZt(\zeroSet\setminus\{n\}, \tfzt+1)\in\cS_{n-1,k}\ ,
  \end{align*}
  where the second line holds by Newton interpolation in \eqref{eq:newtonInterpolation}.
\end{IEEEproof}

\section{Derivation of $\deg_{\betalt}P_{\bT}$}

\label[appendix]{apendix:fieldSizeSkewPoly}
From \eqref{eq:skewPolyEachRow}, the entry $T_{i,j+1}$ in $\bT$ is the coefficient of $\SkewVar^{j}$ in $\rowpoly_i(\SkewVar)$. By \eqref{eq:lclmMinPoly}, it can be seen that $\rowpoly_i(\SkewVar)$ is monic, and therefore $T_{ik}=1$. For $1\leq h<k$, $T_{ih}$ is a commutative multivariate polynomial in the variables $\betalt$'s and
\begin{align*}
  \deg_{\betalt}T_{ih}\leq \deg_{\betalt} \rowpoly_i(\SkewVar)\ . 
\end{align*}
For any $l=1,\dots,\ell$ and $t=1,\dots,n_l$, to find 
$\deg_{\betalt}\rowpoly_i(\SkewVar)$, consider the \newinline{property}
of $\rowpoly_i(\SkewVar)$ in \newinline{\eqref{eq:rootConstraints}}. 
Suppose that $\indMap(l,t)\in\zeroSet_i$, otherwise $\deg_{\betalt}\rowpoly_i(\SkewVar)=0$. Given $j\in\zeroSet_i$, let $(l,t)=\indMap^{-1}(j)$ s.t.~$\alpha_j=a_l\betalt^{q-1}$, then
\begin{align*}
  \rowpoly_i(\SkewVar) = \underset{\alpha\in\{a_{l'}\beta_{l',t'}^{q-1}\ |\ \indMap(l',t')\in\zeroSet_i\}}{\lclm}\{(\SkewVar-\alpha)\}\ .
\end{align*}
Let $\rowpoly_i'\in\FrobPolys$ be the minimal polynomial of $\zeroSet_i'\coloneq \zeroSet_i\setminus\{j\}$, i.e.,
\begin{align*}
  \rowpoly_i'(\SkewVar) = \underset{\alpha\in\{a_{l'}\beta_{l',t'}^{q-1}\ |\ \indMap(l',t')\in\zeroSet_i'\}}{\lclm}\{(\SkewVar-\alpha)\}\ ,
\end{align*}
whose degree in $\SkewVar$ is $\deg_{\SkewVar}\rowpoly_i'(\SkewVar)=|\zeroSet_i'|=k-2$.
Since $j=\indMap(l,t)\not\in\zeroSet_i'$, the coefficients of $\rowpoly_i'(\SkewVar)$ are independent of $\betalt$, i.e., $\deg_{\betalt}\rowpoly_i'(\SkewVar)=0$.

By the remainder evaluation of skew polynomials in \cref{lem:evaluation},
\begin{align*}
  \rowpoly_i'(a_l\betalt^{q-1}) &= \sum_{h=0}^{k-2} f'_{i,h}N_h(a_l\betalt^{q-1})\ ,
\end{align*}
and we have
\begin{align*}
  \deg_{\betalt}\rowpoly_i'(a_l\betalt^{q-1}) &= \deg_{\betalt} N_{k-2}(a_l\betalt^{q-1})\\
                                       &= \deg_{\betalt}(a_l\betalt^{q-1})^{\frac{q^{k-2}-1}{q-1}}\\
                                       &= q^{k-2}-1\ .
\end{align*}
By the Newton interpolation in \eqref{eq:newtonInterpolation}, we can write
\begin{align*}
  \rowpoly_i(\SkewVar) &= \left( \SkewVar - \Frobaut( \rowpoly_i'(a_l\betalt^{q-1}))\cdot a_l\betalt^{q-1}\cdot \left(\rowpoly_i'(a_l\betalt^{q-1})\right)^{-1}\right)\cdot \rowpoly_i'(\SkewVar)\\
   &= \left( \SkewVar - \left( \rowpoly_i'(a_l\betalt^{q-1})\right)^{q-1}\cdot a_l\betalt^{q-1}\right)\cdot \rowpoly_i'(\SkewVar)\\
  &= \SkewVar\cdot  \rowpoly_i'(\SkewVar) -  \left( \rowpoly_i'(a_l\betalt^{q-1})\right)^{q-1}\cdot a_l\betalt^{q-1}\cdot \rowpoly_i'(\SkewVar)\ .
\end{align*}
Since $\deg_{\betalt}\rowpoly_i'(\SkewVar)=0$ and so is $\deg_{\betalt}(\SkewVar\cdot \rowpoly_i'(\SkewVar))$, we have
\begin{align*}
  \deg_{\betalt} \rowpoly_i(\SkewVar) &= (q-1)\cdot \deg_{\betalt} \rowpoly_i'(a_l\betalt^{q-1}) + \deg_{\betalt}(a_l\betalt^{q-1})\\
                               &= (q-1)\cdot (q^{k-2}-1)+(q-1)\\
  &=(q-1)\cdot q^{k-2}\ ,
\end{align*}
for all $l,t$ such that $\indMap(l,t)\in\zeroSet_i$.
Hence, $\deg_{\betalt}T_{ih}\leq \deg_{\betalt}\rowpoly_i(\SkewVar)=(q-1)q^{k-2}, \forall h\in[k-1]$. Then,
\begin{align}
  \deg_{\betalt}P_{\bT} &= \deg_{\betalt} \det \bT \\
                        &\leq \max_{\pi\in\xi_k}\sum_{h=1}^k \deg_{\betalt}T_{\pi(h),h} \\
                        &\leq(k-1)(q-1)\cdot q^{k-2}\ ,\label{eq:degBetaPT}
\end{align}
where $\xi_k$ denotes the set of permutations of $[k]$ and the $(k-1)$ in the last inequality is because $T_{ik}=1$ and hence $\deg_{\betalt}T_{ik}=0$.

\section{Proof of Main Result in \cref{thm:sufficientCond}}

\label[appendix]{sec:proofMainResult}
Denote $f_{\Omega}\coloneqq \gcrd_{i\in\Omega} f_i$. By \ref{p:gcrd}, $f_{\Omega}$ is equal to the minimal polynomial of the set $Z_\Omega \coloneq \bigcap_{i\in\Omega}Z_i$.

We first show the direction \ref{item:equiv1}$\implies$\ref{item:equiv2}. Suppose \ref{item:equiv2} does not hold and w.l.o.g., assume that for $\Omega=\{1,2,\dots,\subRows\}\subseteq[k]$, $k-\deg f_\Omega<\sum_{i\in\Omega} (k-\deg f_i)$. For $i\in \Omega$, let $f_i=q_i\cdot f_\Omega$ for some $q_i\in \FrobPolysn$. Then for $g_1, \dots, g_\subRows\in\FrobPolysn$ such that $\deg(g_i\cdot f_i)\leq k-1$, the equation $\sum_{i\in \Omega} g_i\cdot q_i = 0$ gives a homogeneous linear system of equations in coefficients of the $g_i$'s. The number of variables is at least $\sum_{i\in\Omega}(k-\deg f_i)$ and the number of equations is at most $k-\deg f_\Omega$, which is smaller than the number of variables by the assumption. Therefore, one can find $g_1,\dots, g_\subRows$, not all zero, that solve the linear system of equations, which contradicts \ref{item:equiv1}.

  We then show the direction \ref{item:equiv2}$\implies$\ref{item:equiv1} by induction.
  We do induction on the parameters $(k,\numRows,n)$ considered in lexicographical order $\prec$.

  For the induction basis, when $(k\geq \numRows=1,n\geq 0)$, \ref{item:equiv1} always holds due to \ref{p:noZeroDiv}, i.e., $g_1\cdot f_1 = 0 $ implies $g_1=0$.

For $(k\geq \numRows\geq 2, n=0)$, both \ref{item:equiv1} and \ref{item:equiv2} never hold therefore they are equivalent. Note that $n=0\implies \rowpoly_i=\SkewVar^{\tfzt_i}$ for all $i\in[k]$. For any $f_i=\SkewVar^{\tfzt_i}$ and $f_j=\SkewVar^{\tfzt_j}$ with $\tfzt_i\neq \tfzt_j$ (w.l.o.g.~assuming $\tfzt_i>\tfzt_j$), there exist $g_i=1$ and $g_j=-\SkewVar^{\tfzt_i-\tfzt_j}$ such that $g_if_i+g_jf_j=0$ and hence \ref{item:equiv1} never holds. Suppose $\tfzt_1\leq\tfzt_2$, then for $\Omega=\{1,2\}$, \eqref{eq:equiv2} becomes $k-\tfzt_1\geq (k-\tfzt_1)+(k-\tfzt_2)$, which contradicts with $\deg \rowpoly_i=|\zeroSet_i|+\tfzt_i\leq k-1$. Hence, \ref{item:equiv2} never holds.

  For $(k\geq\numRows\geq 2, n\geq 1)$, we do the induction with the following hypotheses:
  \begin{enumerate}[label={\bfseries H\arabic*}]
  \item \label{H:fromPre}
    Assume that \ref{item:equiv2}$\implies$\ref{item:equiv1} is true for all parameters $(k',\numRows',n')\prec(k,\numRows,n)$.
  \item \label{H:fromHere}
    Take any $f_1,\dots, f_\numRows\in\Skewnk$ for which \ref{item:equiv2} is true for $(k,\numRows,n)$.
  \end{enumerate}

  The logic of the proof is summarized in Figure \ref{fig:proofLogic}.
  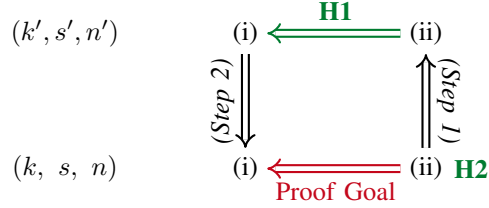
\begin{figure}[h]
    \centering
    \def\x{0.6}
\begin{tikzpicture}
  \node (i) at (0,0) {(i)};
  \node (ii) at ($(i)+(\x*4, 0)$) {(ii)};
  \node (pre) at ($(i)-(\x*4, 0)$) {$(k',s',n')$};
  \node (H1) at ($(i)+(\x*2, \x*0.5)$) {\textcolor{TUMGreenDark}{\ref{H:fromPre}}};

  \node (i2) at ($(i)-(0, \x*3)$) {(i)};
  \node (ii2) at ($(i2)+(\x*4, 0)$) {(ii)};
  \node (here) at ($(i2)-(\x*4, 0)$) {$(k,\ s,\ n)$};
  \node (H2) at ($(ii2)+(\x*1, 0)$) {\textcolor{TUMGreenDark}{\ref{H:fromHere}}};
  \node (Goal) at ($(i2)+(\x*2, -\x*0.5)$) {\textcolor{TUMRed}{Proof Goal}};

  \draw[{Implies}-, thick, double, double distance=\x*3pt, TUMGreenDark] (i.east) -- (ii.west);
  \draw[{Implies}-, thick, double, double distance=\x*3pt] (ii.south) -- (ii2.north);
  \draw[-{Implies}, thick, double, double distance=\x*3pt] (i.south) -- (i2.north);
  \draw[{Implies}-, thick, double, double distance=\x*3pt, TUMRed] (i2.east) -- (ii2.west);

  \node[rotate=-90] (step1) at ($(ii2)+(\x*0.5, \x*1.5)$) {\stepone};
  \node[rotate=90] (step2) at ($(i2)+(-\x*0.5, \x*1.5)$) {\steptwo};
  
\end{tikzpicture}
    \caption{Proof logic for $\ref{item:equiv2}\implies\ref{item:equiv1}$ with initial hypothesis \ref{H:fromPre} and \ref{H:fromHere}.}
    \label{fig:proofLogic}
  \end{figure}

  Starting from \ref{H:fromHere}, we have that for all the subsets $\varnothing\neq \Omega\subseteq[\numRows]$, the inequality \eqref{eq:equiv2} holds.
  We will prove that \ref{item:equiv1} is true for $(k,\numRows, n)$ via \stepone $\rightarrow$ \ref{H:fromPre} $\rightarrow$ \steptwo under different cases:

  \begin{enumerate}[leftmargin=3.4em,label={\bfseries \textit{Case \arabic*}}]
  \item \label{case_s3n2} For $\numRows\geq 3$ and $n\geq 2$,
    \begin{enumerate}[leftmargin=2em, label={\bfseries \textit{Case 1\alph*}}]
    \item \label{case_s3n2_a} $\forall i\in[\numRows]$, $\tfzt_i\geq 1$ (i.e., $|Z_i|\leq k-2$).
      \textit{(In this case, we do induction by reducing $k$.)}
    \item \label{case_s3n2_b} $\exists$ a unique $i\in[\numRows]$ such that $\tfzt_i=0$.
      \textit{(In this case, we do induction by reducing $k$. We may need to reduce $s$ as well.)}
    \item \label{case_s3n2_c} $\exists \Omega \subset [\numRows]$ with $2\leq |\Omega|\leq \numRows-1$ such that \eqref{eq:equiv2} holds with equality.
      \textit{(In this case, we do induction by reducing $s$.)}
    \item \label{case_s3n2_d} $\forall \Omega\subset [\numRows]$ with $2\leq |\Omega|\leq \numRows-1$, \eqref{eq:equiv2} holds strictly and $\exists$ at least two $i\in[\numRows]$ such that $\tfzt_i=0$.
      \textit{(In this case, we do induction by reducing $n$.)}
    \end{enumerate}
  \item \label{case_s2n2} For $\numRows=2$ and $n\geq 2$,
    \begin{enumerate}[leftmargin=2em, label={\bfseries \textit{Case 2\alph*}}]
    \item \label{case_s2n2_a} $\forall i\in\{1,2\}$, $\tfzt_i\geq 1$ (i.e., $|Z_i|\leq k-2$).
      \textit{(The same as \ref{case_s3n2_a}.)}
    \item \label{case_s2n2_b} $\exists$ a unique $i\in\{1,2\}$ such that $\tfzt_i=0$.
      \textit{(The same as \ref{case_s3n2_b}.)}
    \item \label{case_s2n2_c} $\forall i\in\{1,2\}$, $\tfzt_i=0$.
      \textit{(In this case, we do induction by reducing $n$.)} 
    \end{enumerate}
  \item \label{case_s2n1} For $\numRows\geq 2$ and $n=1$,
    \begin{enumerate}[leftmargin=2em, label={\bfseries \textit{Case 3\alph*}}]
    \item \label{case_s2n1_a} $\forall i\in[\numRows]$, $\tfzt_i\geq 1$ (i.e., $|Z_i|\leq k-2$).
      \textit{(The same as \ref{case_s3n2_a}.)}
    \item \label{case_s2n1_b} $\exists$ a unique $i\in\{1,2\}$ such that $\tfzt_i=0$.
      \textit{(The same as \ref{case_s3n2_b}.)}
    \item \label{case_s2n1_c} $\exists$ at least two $i\in[\numRows]$, $\tfzt_i=0$.
      \textit{(We show that this case cannot happen if \ref{item:equiv2} is true for $(k\geq s\geq 2, n=1)$.)}
    \end{enumerate}
  \end{enumerate}
  We illustrate the reduction of $s$ and $n$ of the induction under these cases in \cref{fig:induction_reducing}. We omitted the parameter $k$ for clarity and simplicity, since only $s,n$ are essential in classifying the different cases.
  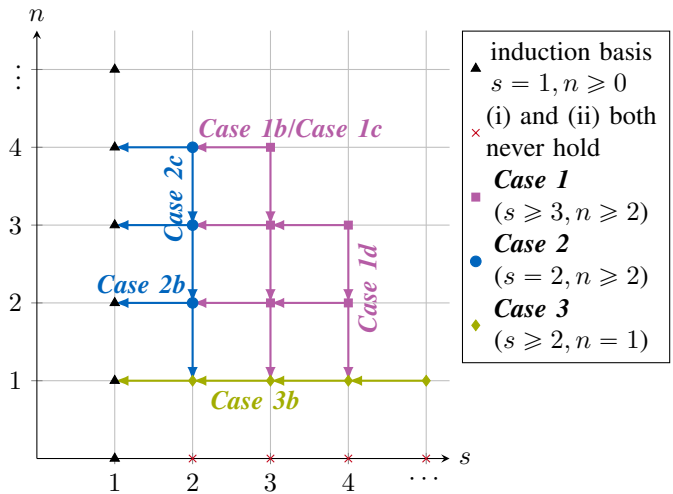
\begin{figure}[htb]
    \centering
    \begin{tikzpicture}

\begin{axis}[
    xmin=0, xmax=5.3,
    ymin=0, ymax=5.5,
    xlabel=$s$,
    ylabel=$n$,
    grid=major,
    xtick={0,1,2,3,4,5},
    xticklabels = {$0$,$1$,$2$, $3$, $4$, $\cdots$},
    ytick={0,1,2,3,4,5,6,7},
    yticklabels={$0$,$1$,$2$, $3$, $4$, $\vdots$},
    unit vector ratio=1 1,
    legend pos=outer north east,
    legend style={cells={align=left}},
    axis lines=middle,
    xlabel style={right},
    ylabel style={above}
]

\addplot[only marks, mark=triangle*, mark size=2pt, color=black] table {
    1 0
    1 1
    1 2
    1 3
    1 4
    1 5
};
\addlegendentry{induction basis\\$s=1,n\geq 0$}

\addplot[only marks, mark=x, mark size=2pt, color=TUMRed] table {
    2 0
    3 0
    4 0
    5 0
};
\addlegendentry{\ref{item:equiv1} and \ref{item:equiv2} both\\never hold}

\addplot[only marks, mark=square*, mark size=1.5pt,color=TUMPink] table {
    3 2
    4 2
    3 3
    4 3
    3 4
};
\addlegendentry{\ref{case_s3n2}\\($s\geq 3,n\geq 2$)}

\addplot[only marks, mark=*, mark size=2pt,color=TUMBlue] table {
    2 2
    2 3
    2 4
};
\addlegendentry{\ref{case_s2n2}\\($s=2, n\geq 2$)}
\addplot[only marks, mark=diamond*, mark size=2pt,color=TUMGreen] table {
    2 1
    3 1
    4 1
    5 1
};
\addlegendentry{\ref{case_s2n1}\\($s\geq 2, n=1$)}

\draw[-latex, thick, color=TUMPink] (3,4) -- (3,3);
\draw[-latex, thick, color=TUMPink] (3,3) -- (3,2);
\draw[-latex, thick, color=TUMPink] (3,2) -- (3,1);
\draw[-latex, thick, color=TUMPink] (4,3) -- node [below, yshift=-2ex, rotate=90] {\ref{case_s3n2_d}} (4,2);
\draw[-latex, thick, color=TUMPink] (4,2) -- (4,1);

\draw[-latex, thick, color=TUMPink] (3,4) -- node [above, xshift=5ex] {\ref{case_s3n2_b}/\ref{case_s3n2_c}} (2,4);
\draw[-latex, thick, color=TUMPink] (4,2) -- (3,2);
\draw[-latex, thick, color=TUMPink] (3,2) -- (2,2);
\draw[-latex, thick, color=TUMPink] (4,3) -- (3,3);
\draw[-latex, thick, color=TUMPink] (3,3) -- (2,3);

\draw[-latex, thick, color=TUMBlue] (2,2) -- node [above, xshift=-1ex] {\ref{case_s2n2_b}} (1,2);
\draw[-latex, thick, color=TUMBlue] (2,3) -- (1,3);
\draw[-latex, thick, color=TUMBlue] (2,4) -- (1,4);
\draw[-latex, thick, color=TUMBlue] (2,4) -- node [above, yshift=-1ex, rotate=90] {\ref{case_s2n2_c}} (2,3);
\draw[-latex, thick, color=TUMBlue] (2,3) -- (2,2);
\draw[-latex, thick, color=TUMBlue] (2,2) -- (2,1);

\draw[-latex, thick, color=TUMGreen] (2,1) -- (1,1);
\draw[-latex, thick, color=TUMGreen] (3,1) -- node [below, xshift=2ex] {\ref{case_s2n1_b}} (2,1);
\draw[-latex, thick, color=TUMGreen] (4,1) -- (3,1);
\draw[-latex, thick, color=TUMGreen] (5,1) -- (4,1);

\end{axis}

\end{tikzpicture}

    \caption{Illustration of the induction for \ref{item:equiv2}$\implies$\ref{item:equiv1} under difference cases.}
    \label{fig:induction_reducing}
  \end{figure}

In the following we demonstrate the induction for each case in detail.

\ref{case_s3n2_a}:
We conveniently denote $k'=k-1$. For all $i\in[\numRows]$, we can write $f_i=\SkewVar\cdot f_i'$, where $f_i'=\fZt(Z_i,\tfzt_i-1)\in \cS_{n,k-1}=\cS_{n,k'}$. Note that since $\min_{i\in[\numRows]} \tfzt_i \geq 1$, we have $\deg f_\Omega\geq 1$ for any $\Omega\subseteq [\numRows]$. For $\Omega = [\numRows]$, \ref{item:equiv2} implies $k-1\geq k-\deg f_{[\numRows]}\geq \sum_{i\in[\numRows]}(k-\deg f_i)\geq s$.

\stepone \ref{item:equiv2} holds for $(f_1',\dots,f_{\numRows}')$ because for any nonempty $\Omega\subseteq [\numRows]$,
\begin{align}
  k'-\deg f_{\Omega}' =& k-\deg f_{\Omega}\nonumber\\
  {\geq}& \sum_{i\in\Omega}(k-\deg f_i)\label{eq:case2astep1}\\
  =&\sum_{i\in\Omega}(k'-\deg f_i')\nonumber
\end{align}
where \eqref{eq:case2astep1} holds because \ref{item:equiv2} holds for $(f_1,\dots, f_\numRows)$ by \ref{H:fromHere}.
By \ref{H:fromPre}, \ref{item:equiv1} then holds for $(f_1',\dots, f_\numRows')\in \cS_{n,k'}$. Note here that we used the induction hypothesis by reducing $k$ to $k'$.

\steptwo We then show that \ref{item:equiv1} also holds for $(f_1,\dots, f_\numRows)$. Suppose that for $g_1,\dots, g_\numRows\in \FrobPolysn$ with $\deg (g_i\cdot f_i)\leq k-1$, we have $\sum_{i=1}^s g_i\cdot f_i=0=\sum_{i=1}^s g_i\cdot (\SkewVar\cdot f_i')\overset{\text{\ref{p:noZeroDiv},\ref{p:xlrdivision}}}{\implies}\sum_{i=1}^{\numRows} g_i f'_i=0$, which implies that $g_1=\cdots=g_\numRows=0$ since \ref{item:equiv1} holds for $(f_1',\dots, f_\numRows')\in\cS^s_{n,k'}$.

\ref{case_s3n2_b}:
Suppose w.l.o.g.~$\tfzt_\numRows=0$ and write $f_\numRows'=f_\numRows\in\cS_{n,k}$.
For $i\in[\numRows-1], \tfzt_i\geq 1$, then we can write $f_i=\SkewVar\cdot f_i'$, where $f_i'=\fZt(Z_i, \tfzt_i-1)\in\cS_{n,k-1}$. Note that $f_\numRows'=f_\numRows\in\cS_{n,k-1}$ if and only if $\deg f_\numRows\leq k-2$, in which case for $\Omega=[s]$, \ref{H:fromHere} implies
\begin{align*}
  k\geq k-\deg f_{\Omega}\geq \sum_{i\in\Omega} (k-\deg f_i) \geq s+1.
\end{align*}

\stepone We show that \ref{item:equiv2} holds for $(f_1',\dots, f_\numRows')$ when $k$ is replaced by $k'=k-1$. First consider the case of $\Omega\subseteq[s-1]$. Since $\forall i\in [s-1],  \tfzt_i\geq 1$, the claim follows similarly to \ref{case_s3n2_a}. Additionally, by the induction hypothesis for $(k'=k-1,s-1,n)$ we get that \ref{item:equiv1} is true for $(f_1',\dots, f_{\numRows-1}')$.
Then consider the case of $\Omega$ such that $s\in\Omega$. Since $f_\numRows=\lclm_{\alpha\in\{a_l\betalt^{q-1} | \indMap(l,t)\in \zeroSet_\numRows\}}\{(\SkewVar-\alpha)\}$ has no factor $\SkewVar$, we have $\gcrd\{f_\numRows, f_i\}=\gcrd\{f_\numRows', f_i'\}, \forall i\in[s-1]$, hence $f_\Omega=f_\Omega'$ where we define $f'_\Omega=\gcrd_{i\in\Omega}\{f'_i\}$. Then
\begin{align}
  k-1-\deg f_{\Omega}'&=-1+k-\deg f_{\Omega}\nonumber\\
&\geq -1+\sum_{i\in\Omega}\deg (k-\deg f_i)\label{eq:case2bstep1}\\
&= k-1-\deg f_\numRows+\sum_{i\in\Omega\setminus\{s\}}(k-\deg f_i)\nonumber\\
&= k-1-\deg f_\numRows'+\sum_{i\in\Omega\setminus\{s\}}(k-1-\deg f_i')\nonumber\\
&=\sum_{i\in\Omega}(k-1-\deg f_i')\ ,\nonumber
\end{align}
where \eqref{eq:case2bstep1} holds from \ref{H:fromHere}.
By \ref{H:fromPre}, \ref{item:equiv2} $\implies$ \ref{item:equiv1} is true for $(f_1',\dots, f_\numRows')$ with parameters $(k'=k-1,s,n)$ if $\deg f_\numRows'\leq k-2$., which implies $k\geq s+1$.

\steptwo Suppose that for some $g_1,\dots, g_s\in\FrobPolysn$ with $\deg(g_i\cdot f_i)\leq k-1$ we have $\sum_{i=1}^s g_i\cdot f_i=0$.
Then $0 = \sum_{i=1}^s g_i\cdot f_i = g_s\cdot f_s + \sum_{i=1}^{s-1}g_i\cdot (\SkewVar \cdot f_i' )$, which implies $\SkewVar|(g_s\cdot f_s)$. However, since $\SkewVar\not\divides f_s$, by \ref{p:xdivision}, $\SkewVar| g_s$. Then we can write $g_s=g_s'\cdot \SkewVar$ for some $g_s'\in\FrobPolysn$ with $\deg g_s'=\deg g_s-1$.

If $\deg f_s=k-1$, then $\deg g_s'=-1$, implying $g_s=0$. Since \ref{item:equiv1} holds for $(f_1',\dots,f_{s-1}')\in \cS_{n,k-1}^{s-1}$ with the parameter tuple $(k-1,s-1,n)$, $g_1,\dots, g_{s-1}$ are also zero.
Note that here we used the induction hypothesis by reducing $k$ to $k-1$ and $s$ to $s-1$.

If $\deg f_s \leq k-2$, we have
\begin{align*}
  0 =  &  \sum_{i=1}^s g_i\cdot f_i\\
  = & (g_s'\cdot \SkewVar)\cdot f_s+\sum_{i=1}^{s-1}g_i\cdot (\SkewVar\cdot f_i')\\
  = & (g_s'\cdot \SkewVar)\cdot f_s' +\sum_{i=1}^{s-1} (g_i\cdot \SkewVar)\cdot f_i'\ .
\end{align*}
Then $g_1 = \cdots = g_{s-1}=g_s' =0 $ since \ref{item:equiv1} holds for $(f_1',\dots, f_s')\in \cS_{n,k-1}^{s}$ with the parameter tuple $(k-1,s,n)$.
Note that here we used the induction hypothesis by reducing $k$ to $k-1$.
Hence, all $g_1=\cdots = g_s = 0$.

\ref{case_s3n2_c}:
W.l.o.g., assume that \eqref{eq:equiv2} holds with equality for $\Omega'=\{1,\dots,\subRows\}$, $1< \subRows<\numRows$, i.e.,
\begin{align}
  \label{eq:case1-assumption}
  k-\deg f_0 = \sum_{i\in\Omega'}(k-\deg f_i)\ ,
\end{align}
where $f_0=f_{\Omega'}=\gcrd_{i\in\Omega'}f_i$. Since $f_0|_r f_i, \forall i\in\Omega'$, there exists $f_i'\in\FrobPolysn$ such that $f_i=f_i'\cdot f_0$.
Since $\subRows<\numRows$ and $\numRows-\subRows +1 <\numRows$, we split $(f_1,\dots, f_\numRows)\in \Skewnk^{\numRows}$ into two smaller problems $(f_1,\dots,f_\subRows)\in\Skewnk^{\subRows}$ with the parameter tuple $(k,\subRows<\numRows,n)$ and $(f_0, f_{\subRows+1},\dots, f_{\numRows})\in\Skewnk^{\numRows-\subRows+1}$ with the tuple $(k,\numRows-\subRows+1<\numRows, n)$.

\stepone Note that by \ref{H:fromHere}, \ref{item:equiv2} is true for $(f_1,\dots,f_{\subRows})$ and for $(f_0, f_{\subRows+1},\dots, f_{\numRows})$ when $0\not\in\Omega''\subseteq\{0,\subRows+1,\dots,\numRows\}$.
We show in the following that \ref{item:equiv2} is also true for $(f_0, f_{\subRows+1},\dots, f_{\numRows})$ with $0\in\Omega''$:
\begin{align}
  k-\deg f_{\Omega''} \nonumber
  =& k- \deg \gcrd\{f_0, f_{\Omega''\setminus\{0\}}\}\nonumber\\
  =& k- \deg \gcrd\{f_{\Omega'}, f_{\Omega''\setminus\{0\}}\}\nonumber\\
  =& k- \deg \gcrd_{i\in\Omega'\cup \Omega''\setminus\{0\}} f_i\nonumber\\
  {\geq} & \sum_{i\in\Omega'\cup \Omega''\setminus\{0\}} (k-\deg f_i)\label{eq:case1step1}\\
  =& \sum_{i\in\Omega'}(k-\deg f_i) + \sum_{i\in\Omega''\setminus\{0\}}(k-\deg f_i)\nonumber\\
    =& k- \deg f_0 + \sum_{i\in\Omega''\setminus\{0\}}(k-\deg f_i)\label{eq:case1-equal}\\
  =& \sum_{i\in\Omega''}(k-\deg f_i)\nonumber
\end{align}
Note that $\Omega'\cup\Omega''\setminus\{0\}$ is a subset of $\Omega$. Therefore, the inequality in \eqref{eq:case1step1} follows from \ref{H:fromHere}. The equality \eqref{eq:case1-equal} follows from \eqref{eq:case1-assumption}. 
Now we can conclude that \ref{item:equiv2} is true for $(f_1,\dots,f_{\subRows})$ and for $(f_0, f_{\subRows+1},\dots, f_{\numRows})$.

By \ref{H:fromPre}, \ref{item:equiv1} is true for both smaller problems $(f_1,\dots,f_\subRows)\in\Skewnk^{\subRows}$ and $(f_0, f_{\subRows+1},\dots, f_{\numRows})\in\Skewnk^{\numRows-\subRows+1}$.

\steptwo Then we show \ref{item:equiv1} is also true for $(f_1,\dots, f_{\numRows})$. Suppose that for some $g_1,\dots, g_{\numRows}\in\FrobPolysn$ with $\deg g_i\cdot f_i\leq k-1, \forall i\in [\numRows]$, we have
\begin{align}
  \sum_{i=1}^{\numRows} g_i\cdot f_i=0\ . \label{eq:LHSequiv1}
\end{align}
Since $f_0|_r f_i$ for all $i\in \Omega'=[\subRows]$, $f_0$ is a right factor $\sum_{i=1}^{\subRows} g_i\cdot f_i$ and we can then write $\sum_{i=1}^{\subRows}g_i\cdot f_i=g_0\cdot f_0$, for some $g_0\in\FrobPolysn$. Then
\begin{align}
  0 =& \sum_{i=1}^{\numRows} g_i\cdot f_i\nonumber \\
  =& \sum_{i=1}^{\subRows} g_i\cdot f_i + \sum_{i=\subRows+1}^{\numRows} g_i\cdot f_i\nonumber \\
  =& g_0\cdot f_0 + \sum_{i=\subRows+1}^{\numRows} g_i\cdot f_i\label{eq:0equal}
\end{align}
From the conclusion that \ref{item:equiv1} is true for $(f_0,f_{\subRows+1}, \dots, f_{\numRows})$, \eqref{eq:0equal} holds only if $g_0=g_{\subRows+1}=\dots = g_{\numRows}=0$.
Similarly, since \ref{item:equiv1} is true for $(f_1,\dots, f_{\subRows})$, $0=g_0\cdot f_0=\sum_{i=1}^{\subRows}g_i\cdot f_i$ only if $g_1 = \dots = g_{\subRows}=0$.
Therefore, \eqref{eq:LHSequiv1} holds only if $g_1 = \dots = g_{\numRows}=0$ and \ref{item:equiv1} is proven for $(f_1,\dots, f_{\numRows})\in\cS_{n,k}^{\numRows}$ with the parameter tuple $(k,s,n)$.

\ref{case_s3n2_d}:
Assume w.l.o.g.~that $\tfzt_{\numRows-1}=\tfzt_{\numRows}=0$. Then for $i=\numRows-1, \numRows$, $\deg f_{i}=|\zeroSet_{i}|$.
If $\zeroSet_{\numRows-1}=\zeroSet_{\numRows}$, then for $\Omega = \{\numRows-1,\numRows\}$, \ref{item:equiv2} implies
\begin{align*}
  k-\deg f_\numRows=&k-\deg f_{\numRows-1}\\
  =& k-\deg \gcrd \{f_{\numRows-1}, f_{\numRows}\}\\
  \geq& k-\deg f_{\numRows-1}+k-\deg f_{\numRows}
\end{align*}
which contradicts with $\deg f_{i}\leq k-1$ for any $i\in[\numRows]$. Hence,  $\zeroSet_{\numRows-1}\neq[n]$ or $\zeroSet_{\numRows}\neq [n]$.
W.l.o.g., assume $\zeroSet_{\numRows}\neq  [n]$ and $n\not\in \zeroSet_{\numRows}$.

Note that $n=\indMap(\ell,n_{\ell})$. We will substitute the variable $\beta_{\ell,n_\ell}=0$. For all $i\in[\numRows]$, let $f_i'\coloneqq f_i|_{\beta_{\ell,n_\ell}=0}$. Since $n\not\in\zeroSet_s$, we have $f_s'=f_s\in\cS_{n-1,k}$. For other $i\in[\numRows-1]$, by \ref{p:reducevar}, $f_i'\in\cS_{n-1,k}$ and
\begin{align}\label{eq:fi2c}
  f_i' = \begin{cases}
    \fZt(Z_i,\tfzt_i) & n\not\in Z_i\\
    \fZt(Z_i\setminus\{n\}, \tfzt_i+1) & n\in Z_i
  \end{cases}\ .
\end{align}
In the first case of~\eqref{eq:fi2c} we denote $\zeroSet_i' = \zeroSet_i$ and $\tau'_i=\tau_i$, whereas in the second we denote $\zeroSet_i' = \zeroSet_i\setminus\{n\}$ and $\tau'_i=\tau_i+1$. Additionally, we define $f'_\Omega=\gcrd_{i\in\Omega}f'_i$.

\stepone We will first show that $(f_1',\dots, f_s')$ satisfies \ref{item:equiv2}. That is, we show that $\forall \varnothing\neq \Omega'\subseteq[s]$, $k-\deg f_{\Omega'}'\geq  \sum_{i\in\Omega'} (k-\deg f_i') $.

For $|\Omega'|=1$, it is trivial.

For $2\leq |\Omega'|\leq s-1$,
\begin{align}
  k-\deg f_{\Omega'}' =& k-|\bigcap_{i\in\Omega'}\zeroSet_i'|-\min_{i\in\Omega'} \tfzt_i'\nonumber \\
  \geq& k- |\bigcap_{i\in\Omega'}\zeroSet_i|  - \min_{i\in\Omega'} \tfzt_i - 1 \label{eq:case2cineq1}\\
  =& k-\deg f_{\Omega'} -1\nonumber \\
  \geq & \sum_{i\in\Omega'} (k-\deg f_i) \label{eq:case2cineq2}\\
  =& \sum_{i\in\Omega'} (k-\deg f_i')\label{eq:case2c-deg-equal}\ .
\end{align}
The inequality \eqref{eq:case2cineq1} is because $|\bigcap_{i\in\Omega}\zeroSet_i'|\leq |\bigcap_{i\in\Omega} \zeroSet_i|$ and $\min_{i\in\Omega} \tfzt_i'\leq \min_{i\in\Omega} \tfzt_i +1$. The inequality \eqref{eq:case2cineq2} is because we assume the inequality \eqref{eq:equiv2} in \ref{item:equiv2} holds strictly for all $2\leq |\Omega|\leq s-1$. The equality \eqref{eq:case2c-deg-equal} holds because $\deg f_i'=\deg f_i,\forall i\in[\numRows]$ by observing \eqref{eq:fi2c}.

For $|\Omega'|=s$, \eqref{eq:equiv2} is not necessarily strict. However, since
\begin{align*}
n\not\in \zeroSet_s &\implies n\not\in \bigcap_{i\in[s]}\zeroSet_i\\
&\implies |\bigcap_{i\in[\numRows]}\zeroSet_i'|=|\bigcap_{i\in[\numRows]} \zeroSet_i|\implies f'_{[\numRows]}=f_{[\numRows]},
\end{align*}
we have
\begin{align*}
  k-\deg f'_{[\numRows]} =& k-\deg f_{[\numRows]}\\
  \geq & \sum_{i\in[\numRows]} (k-\deg f_i)\\
  = & \sum_{i\in[\numRows]} (k-\deg f_i')\ .
\end{align*}

Hence, \ref{item:equiv2} holds for $(f_1',\dots, f_\numRows')\in\cS_{n-1,k}^{\numRows}$.
By \ref{H:fromPre}, \ref{item:equiv1} holds for $(f_1',\dots, f_\numRows')\in\cS_{n-1,k}^{\numRows}$ with the parameter tuple $(k\geq\numRows\geq 3,n-1)$ where $n\geq 2$. Note that here we used the induction hypothesis by reducing $n$ to $n-1$.

\steptwo Suppose that for some $g_1,\dots, g_s\in\FrobPolysn$, not all zero, with $\deg (g_i\cdot f_i)\leq k-1$, we have $\sum_{i=1}^s g_i\cdot f_i =0$. Let $g_i'=g_i|_{\beta_{\ell,n_\ell}=0}\in R_{n-1}[\SkewVar;\Frobaut]$. Further assume that at least one coefficient of some $g_i$ is not divisible by $\beta_{\ell,n_\ell}$ (otherwise, divide them by $\beta_{\ell,n_\ell}$). Then $g_i'$ are not all zero. We can write
\begin{align*}
  \sum_{i=1}^s g_i'\cdot f_i' = \left(\sum_{i=1}^s g_i\cdot f_i \right)|_{\beta_{\ell,n_\ell}=0} = 0|_{\beta_{\ell,n_\ell}=0} = 0\ .
\end{align*}
However, this contradicts \ref{item:equiv1} being true for $(f_1',\dots, f_s')$ with the parameter tuple $(k,s,n-1)$. Therefore, $g_1,\dots, g_s\in\FrobPolysn$ must be all zero to have $\sum_{i=1}^s g_i\cdot f_i =0$.

\ref{case_s2n2_c}: In this case we have $\Omega=\{1,2\}$ and $\tfzt_1=\tfzt_2=0$.
Similar to \ref{case_s3n2_d}, $\zeroSet_1\neq [n]$ or $\zeroSet_2\neq[n]$. W.l.o.g., assume $\zeroSet_2\neq[n]$ and $n\not \in \zeroSet_2$. Note that $n=\indMap(\ell,n_{\ell})$. We substitute the variable $\beta_{\ell,n_{\ell}}=0$. For $i=1,2$, let $f_i'\coloneqq f_i|_{\beta_{\ell,n_\ell}=0}$ and $f'_{\Omega}\coloneq \gcrd\{f_1', f_2'\}$.
Since $n\not\in\zeroSet_2$, $f_2'=f_2$.
By \ref{p:reducevar}, $f_1'\in\cS_{n-1,k}$ and
\begin{align*}
      f_1' = \begin{cases}
    \fZt(Z_1,0) & n\not\in Z_1\\
    \fZt(Z_1\setminus\{n\}, 1) & n\in Z_1
    \end{cases}\ .
  \end{align*}

\stepone We first show that $(f_1',f_2')\in\cS_{n-1,k}^2$ satisfies \ref{item:equiv2}. That is, we show that $\forall \varnothing\neq \Omega'\subseteq\Omega$, $k-\deg f_{\Omega'}'\geq  \sum_{i\in\Omega'} (k-\deg f_i') $.

For $|\Omega'|=1$, it is trivial.

For $\Omega'=\{1,2\}$, since
\begin{align*}
    n\not\in\zeroSet_2\implies n\not\in \zeroSet_1\cap\zeroSet_2 &\implies |\zeroSet_1'\cap\zeroSet_2'|=|\zeroSet_1\cap\zeroSet_2|\\
    &\implies \deg f'_{\Omega'}= \deg f_{\Omega'},
\end{align*}
we have
\begin{align*}
  k-\deg f'_{\Omega'} =& k-\deg f_{\Omega'}\\
  \geq & \sum_{i\in\Omega'} (k-\deg f_i)\\
  = & \sum_{i\in\Omega'} (k-\deg f_i')\ .
\end{align*}
Hence, \ref{item:equiv2} holds for $(f_1',f_2')\in\cS_{n-1,k}^2$.
By \ref{H:fromPre}, \ref{item:equiv1} holds for $(f_1',f_2')\in\cS_{n-1,k}^2$ with parameter tuple $(k\geq \numRows=2, n-1)$ where $n\geq 2$. Here we used the induction hypothesis by reducing $n$ to $n-1$.

\steptwo This step can be shown in the same manner as in \ref{case_s3n2_d}.

\ref{case_s2n1_c}: W.l.o.g., assume that $\tfzt_1=\tfzt_2=0$. Since $n=1$, $|\zeroSet_i|\leq n=1, \forall i\in [\numRows]$. By the definition of $f_i\in \cS_{1,k}$ in \eqref{eq:Skewnk}, $\deg f_i \leq n=1$. Assume \ref{item:equiv2} is true for this case, then for $\Omega=\{1,2\}$, we have
\begin{align}
  k-\deg f_{\Omega}\geq k-\deg f_1 + k-\deg f_2. \label{eq:s2n1_equiv2}
\end{align}
If $\zeroSet_1=\zeroSet_2=\varnothing$ or $\{1\}$, then $\deg f_{\Omega}=\deg f_1=\deg f_2\leq n=1$ and \eqref{eq:s2n1_equiv2} implies $\deg f_1\geq k$, which contradicts $k\geq s\geq 2$.
Otherwise, w.l.o.g., assume $\zeroSet_1=\varnothing$ and $\zeroSet_2 = \{1\}$, then $\deg f_{\Omega}=0$ and \eqref{eq:s2n1_equiv2} implies $1=\deg f_2\geq k$, which contradicts $k\geq s\geq 2$.
Therefore, if \ref{item:equiv2} is true for $(k\geq s\geq 2,n=1)$, this case cannot happen.

\bibliographystyle{ieeetr}
\bibliography{refs}
\end{document}